\documentclass[12pt]{article}
\usepackage[latin1]{inputenc}
\usepackage[a4paper,left=2.5cm,right=2.5cm,top=2cm,bottom=1.5cm,footnotesep=.75cm,headheight=13.6pt]{geometry}
\usepackage{setspace}
\usepackage{amsthm}
\usepackage{amsmath}
\usepackage{amsfonts}
\usepackage{amssymb}
\usepackage{graphicx}
\usepackage[small]{caption}
\usepackage{authblk}
\usepackage{rotating}
\usepackage{float}
\usepackage{pdflscape}
\usepackage{subcaption}
\usepackage{longtable}
\usepackage{xcolor}
\usepackage{lscape}
\usepackage{endnotes}
\newtheorem{theorem}{Theorem}[section]

\newtheorem{lemma}{Lemma}[section]

\title{Robust estimation of the autocorrelation function via forward ratios.}

\author[1]{\large Antonio Monta\~{n}\'es\thanks{Corresponding author. E-mail: amontane@unizar.es}}
\author[2]{\large Esther Ruiz}
\affil[1]{Dpt. of Economic Analysis, Universidad de Zaragoza (Spain)}
\affil[2]{Department of Statistics, Universidad Carlos III de Madrid (Spain)}

\date{\today}

\begin{document}

\maketitle

\begin{abstract}

It is obvious to say that an adequate estimation of the autocorrelation function is central in time series analysis. In this paper, we propose three new robust estimators based on ratios of observations, which offer strong resistance against outliers. While the first estimator, which is based on the median, is not efficient, the second is a Quasi Maximum Likelihood (QML) estimator with better efficiency properties. The third estimator is a plug-in estimator, which does not require numerical optimization and, consequently, is extremely simple from a computationally point of view, having similar efficiency to that of the ML estimator. We derive the asymptotic distribution of the first two estimators, when the true autocorrelations are zero. Furthermore, we also show that the asymptotic distribution of the plug-in estimator is rather close to that of the QML estimator, allowing for inference and, in particular, for the construction of point-wise significance bands for the autocorrelations. Using Monte Carlo simulations, we analyse the finite sample properties of the proposed estimators and compare them with those of the sample autocorrelations and alternative extant robust estimators based on ranks. Although the proposed estimators have larger dispersion than the sample autocorrelations in uncontaminated time series, they are highly robust in the presence of outliers. Also, they have better properties than popular alternative robust estimators based on ranks when estimating autocorrelations of order larger than one. The results are illustrated by estimating the correlogram of daily IBEX35 returns, quarterly US economic growth and monthly US inflation.
\end{abstract}



\setcounter{page}{1}

\newpage

\doublespacing

\section{Introduction}

Modelling temporal dependence is central to the analysis of time series data, with estimation and inference of the autocorrelation function being a popular first step; see, for example, the Royal Society address by Yule (1921) and Bartlett (1935), for seminal references. Furthermore, popular diagnostic tools for time series models are based on testing for lack of autocorrelation of their residuals; see, for example, Box and Pierce (1970) and Durbin (1970) for early references. The autocorrelations have also been used as a criteria for time series clustering; see, for example, Albino, Caiado and Crato (2024) and Caiado and Crato (2026). 

The most popular estimator of the autocorrelation function is its sample counterpart, which is well known to be non-robust in presence of extreme unexpected observations that look discordant from most other observations (outliers), which often appear in the empirical analysis of real time series. Outliers may cause serious bias in estimating autocorrelations. One can follow two main strategies to deal with outliers in real time series. First, outliers can be identified and their effects removed from observations before estimating the autocorrelation function; see, for example, Chang, Tiao and Chen (1988). Alternatively, robust estimators of the autocorrelations have been proposed, with most of them classified into two large classes. The first class contains estimators based on the Method of Moments (MM); see, for example, Chang and Politis (2016). These estimators usually require large computational effort, limiting their empirical implementation. Second, several robust estimators are based on order statistics and, consequently, are less computationally intensive; see, for example, the popular procedure proposed by Ma and Genton (2000). The main limitation of rank-based procedures is that their asymptotic distribution is unknown and, consequently, inference is restrained. Furthermore, we will see that, even though they have nice finite sample properties when estimating the first order autocorrelation, their performance worsen when estimating autocorrelations of larger order.

Alternatively, in the context of the linear first order autoregressive, AR(1), model with known mean, Hurwicz (1950) proposes a computationally very simple estimator of the autoregressive parameter (first order autocorrelation), based on the median of ratios of consecutive observations; see also Guo and Billard (2012), who consider the estimator of Hurwicz for one-step-ahead predictions. Being based on the median, the estimator proposed by Hurwicz (1950) is median unbiased; see Zieli\'nski (1999). However, it  is not efficient and, consequently, Reschenhofer (2019) proposes two alternative estimators with increased efficiency. The first estimator of the autoregressive parameter is a Quasi-Maximum Likelihood (QML) estimator modified to guarantee that the ratio of consecutive observations are restricted to lye in the interval $[-1,1]$. The second estimator is based on the mean of the modified ratios of consecutive observations.

In this paper, we extend the estimators proposed by Reschenhofer (2019) to estimate not only the first order autocorrelation but the whole autocorrelation function of stationary linear autoregressive-moving average (ARMA) models with the innovations having a generic symmetric distribution (non-necessary Gaussian). The proposed estimators of the autocorrelation function are designed to improve their finite sample properties in the presence of large absolute autocorrelations not only when the series may be contaminated by large additive outliers but also when the mean is unknown. Using Monte Carlo simulations, we show that the finite sample performance of the proposed estimators have robustness properties comparable to those of the robust estimator proposed by Ma and Genton (2000) when estimating the first order autocorrelation, and better properties when estimating autocorrelations of larger lags.

When the mean is known, we derive the asymptotic distribution of the proposed QML estimator under the null of uncorrelatedness, allowing for the construction of the popular pointwise significance bands for the estimated autocorrelations and of portmanteau tests for the joint significance of a set of consecutive autocorrelations. We use simulated data to show that, regardless of whether the mean is know or not, the finite sample approximation of the asymptotic distribution is adequate to represent the empirical distribution of the estimated autocorrelations.

Finally, the relevance of using robust estimators of the autocorrelations is illustrated by estimating the autocorrelation function of three real time series. First, we consider daily returns of the IBEX35 index of the Madrid Stock Exchange. In this case, we show that when the autocorrelations are estimated using their sample counterparts, several of them are significant, providing evidence against the Efficient Markets Hypothesis (EMH). However, by using the robust plug-in estimator proposed in this paper, there is not any significant autocorrelation, and, therefore, there is not evidence of predictability of daily IBEX35 price movements. The second empirical application considered illustrates the contrary effect. We estimate the autocorrelations of quarterly US economic growth. When the autocorrelations are estimated using the standard sample autocorrelations, the evidence of linear temporal dependence is very weak. However, when the proposed robust estimator is used, this evidence is stronger. Finally, we estimate the autocorrelations of monthly inflation in the US, which is subject to a debate about whether it is better represented by a stationary, long-memory or non-stationary process. We show that when the robust-plug.in estimator of the autocorrelations is used instead of the standard sample autocorrelations, the estimated autocorrelations are larger in magnitude and more persistent, and, consequently, the evidence to distinguish between long-memory and non-stationarity is weaker. 

The rest of the paper is organized as follows. Section \ref{sec:prelim} describes preliminaries, introducing the limiting distribution of the sample autocorrelations, the popular robust estimator of the autocorrelation function proposed by Ma and Genton (2000), and the estimator based on ratios of consecutive observations proposed by Hurwicz (1950) and modified by Reschenhofer (2019). In Section \ref{sec:ratios}, we introduce the estimators proposed in this paper and derive their properties and asymptotic distribution. Section \ref{sec:simu} presents the results of an exhaustive simulation study. Section \ref{sec:empiric} describes the empirical illustration with the estimation of the autocorrelations of IBEX35 daily returns, of quarterly US economic growth, and of monthly US inflation. Finally, Section \ref{sec:conclusions} concludes with a summary of the main results.
  
\section{Autocorrelation function: Estimation and inference}
\label{sec:prelim}

In this section, we describe the main properties of the sample autocorrelations and of robust estimators based on ranks and on ratios of consecutive observations.

\subsection{Sample autocorrelations}

Consider the following stationary process
\begin{equation}
\label{eq:process}
y_t=\mu + \sum_{i=0}^{\infty} c_i \varepsilon_{t-i},
\end{equation}
where $t=1,2,...,T$, $c_0=1$, $ \sum_{i=0}^{\infty} c_i^2 < \infty$, and $\varepsilon_t$ is white noise with variance $\sigma^2_{\varepsilon}$.  Then, $y_t$ has mean $\mu$, finite variance $\sigma^2_y= \sigma^2_{\varepsilon} \sum_{i=0}^{\infty} c_i^2 < \infty$, and, for $h \in \mathbb{Z}$, autocorrelation between $y_t$ and $y_{t-h}$, given by $\rho(h)=\frac{\gamma(h)}{\sigma^2_y}$, where $\gamma(h)$ is the autocovariance between $y_t$ and $y_{t-h}$, which is a function of $c_i$.

Consider first that the mean, $\mu$, is known and, without loss of generality, equal to zero. In this case, given observations, $y_t$, $t=1,...,T$, the autocorrelation of order $h$, $\rho(h)$, can be estimated by the corresponding sample autocorrelation as follows:
\begin{equation}
\label{eq:sample_1}
r^*(h)=\frac{\sum_{t=h+1}^{T} y_t y_{t-h}}{\sum_{t=1}^T y_t^2}.
\end{equation}

The estimator of the autocorrelations in (\ref{eq:sample_1}) has well known properties; see Bartlett (1946), who derives its expectation and asymptotic distribution. However, it is of limited practical relevance given that in most empirical applications, the mean, $\mu$, is unknown and needs to be estimated. In this case, the sample autocorrelation of order $h$ is given by
\begin{equation}
\label{eq:sample_2}
r(h)=\frac{\sum_{t=h+1}^{T} \left(y_t - \bar{y} \right) \left( y_{t-h} - \bar{y} \right)}{\sum_{t=1}^T \left(y_t - \bar{y} \right)^2},
\end{equation}
where $\bar{y}=T^{-1} \sum_{t=1}^T y_t$. The graphical representation of the sample autocorrelations, $r(h)$,  for $h=1,...,H$, is often known as correlogram. The estimator in (\ref{eq:sample_2}) is significantly downward biased especially when the persistence is large; see the discussion by Andrews (1993) and the references therein for the related problem of the Least Squares (LS) estimation of the autoregressive parameter of an AR(1) model, and Vogelsang and Yang (2016), who propose an alternative unbiased estimator of the autocorrelations based on the linear relationship between the vector of sample autocovariances and the vector of their population counterparts.\endnote{Nielsen (2006), shows that $r(h)$ could deliver misleading inferences in non-stationary frameworks.}

In practice, the sample correlogram is usually accompanied by pointwise non-rejection bands for the null hypothesis of lack of temporal autocorrelation; see, for example, Hassler, Pohle and Zahn (2025), who call them ``significance bands''. These pointwise significance bands are described in basically every textbook on time series analysis and are added by default to correlograms in most statistical software; see, for example, Sri Ranganath (2018) for an empirical application using the correlogram with pointwise significance bands, and Caiado and Crato (2026), who use them to identify highly significant autocorrelations to be used in discrepancy statistics for clustering time series. They are based on the asymptotic distribution of the sample autocorrelations. Assume independent and identically distributed (\textit{i.i.d.}) disturbances, $\varepsilon_t$, and $\sum_{i=0}^{\infty} i c_i < \infty$ and consider the vector of consecutive sample autocorrelations, $\left( r(1),...,r(H) \right)^{\prime}$. As $T \rightarrow \infty$, Anderson and Walker (1964) show that\endnote{The classical restrictions for the asymptotic distribution in (\ref{eq:asymp_rho}) rely only on finite second order moments. The asymptotic distribution of the sample autocorrelations were previously derived by other authors under stronger assumptions on the existence of moments; see, for example Bartlett (1946).}
\begin{equation}
\label{eq:asymp_rho}
\sqrt{T} \left( r(1)-\rho(1),...,r(H)-\rho(H)\right)^{\prime} \xrightarrow{d} N \left( \mathbf{0}, \mathbf{W} \right),
\end{equation}
where $\mathbf{W}$ is the Bartlett covariance matrix with representative element given by
\begin{equation}
\label{eq:asymt_2}
\omega_{ij} = \sum_{\nu=0}^{\infty} \left[ \rho(\nu+i) \rho(\nu-i) - 2 \rho(\nu) \rho(i) \right] \left[\rho(\nu +j) \rho(\nu-j) - 2 \rho(\nu) \rho(j) \right].
\end{equation}

From (\ref{eq:asymt_2}), the asymptotic variance of $r(h)$ can be obtained when $i=j$ as follows:
\begin{equation}
\label{eq:var_rho}
Var(r(h))= \frac{1}{T} \sum_{j=1}^{\infty} \left(\rho(h+j) + \rho (h-j) -2 \rho(h) \rho(j)\right)^2;
\end{equation}
see also Phillips and Solo (1992). Another important aspect of the asymptotic covariance in (\ref{eq:asymp_rho}) is that it depends only on the autocorrelations $\rho(h)$ themselves and not on fourth order moments, as it is the case for the sample autocovariances. Hannan and Hyde (1972) relax the \textit{i.i.d.} assumption on $\varepsilon_t$ and show that asymptotic normality remains valid under some additional regularity conditions on the innovations.\endnote{Recently, Dalla, Giaritis and Phillips (2022) adapt the standard correlogram and portmanteau tests to accomodate nonstationarity, conditional heteroscedasticity and cross-correlation at various lags in bivariate time series.}

Under the null, $\rho(h)=0$ for $h \neq 0$, $y_t$ is \textit{i.i.d.}, and $\mathbf{W}=\mathbf{I}_H$. In this case, the finite sample approximation of the asymptotic variance of $r(h)$ is $\frac{1}{T}$, and consequently, the popular Bartlett pointwise band with confidence $(1-\alpha)\%$ for $\rho(h)$ is given by
\begin{equation}
\label{eq:bounds}
\pm \frac{z_{0.5 \alpha}}{\sqrt{T}},
\end{equation}
where $z_{x}$ is the $x$ quantile of the standard normal distribution; see del Barrio Castro, Escribano and Sibbertsen (2025) and Gospodinov, Lopez Gaffney and Ng (2025) for some few selected empirical applications using these bounds.

Consider now testing for the joint significance of a set of consecutive autocorrelations. In particular, for $H \ge 1$, consider testing the null hypothesis $H_0: \rho(1)=...=\rho(H)=0$, against the alternative $H_1: \rho(1)= \rho(H-1)=0, \rho(H) \neq 0$. It is often suggested to consider $H=\min\left(10, \frac{T}{5}\right)$.  It is popular to use the following portmanteau test proposed by Box and Pierce (1970)\endnote{The BP test was originally proposed to test for uncorrelatedness of residuals of time series models. However, it is common to report it together with the sample autocorrelations in most time series software packages; see also Lobato, Nankervis and Savin (2001) and Kwan and Sim (1996), who consider testing for uncorrelatedness of the original time series.}
\begin{equation}
\label{eq:BP}
Q_{BP}(H)=T \sum_{i=1}^H  r(i)^2.
\end{equation}
The BP statistic in (\ref{eq:BP}) was modified by Ljung and Box (1978) to have better finite sample properties as follows
\begin{equation}
\label{eq:LB}
Q_{LB}(H)=T(T+2) \sum_{i=1}^H \left( T-i \right)^{-1} r(i)^2.
\end{equation}
Under independence of $y_t$, $Q_{BP}(H)$ and $Q_{LB}(H)$ are asymptotically equivalent and distributed as a $\chi^2$ with $H$ degrees of freedom.

In spite of the popularity and good properties of the sample autocorrelations in (\ref{eq:sample_2}), it is well known that they are very sensitive to the presence of outliers, which may cause severe biases and hence seriously jeopardize their function as model identification and diagnostic tools; see, for example, Guttman and Tiao (1978), Miller (1980), Tsay (1986), and Fellag and Zieli\'nski (1996). Consider, for example, the estimator of the autocorrelations in (\ref{eq:sample_1}) and the series $y_t^*=y_t + \delta I(t=t^*)$, where $I(\cdot)$ is the indicator function that takes value 1 when the argument is true and zero otherwise. In this case, the sample autocorrelation of order $h$ is obtained with the contaminated series as follows:
\begin{equation}
\label{eq:sample_3}
r^*(h)=\frac{\sum_{t=h+1}^{T} y_t y_{t-h} + \delta (y_{t^*-h}+ y_{t^*+h})}{\sum_{t=1}^T y_t^2 + \delta^2 + 2 \delta y_{t^*}},
\end{equation}
which goes to zero as $\delta$ increases. Therefore, the sample autocorrelations are highly sensitive to outliers with their significance bands being distorted and the tests for autocorrelation based on them having size and power equal to zero when the size of the outlier is large; see, for example, the discussion by Berkoun, Fellag and Zieli\'nski (2003) in the context of testing for serial correlation in an AR(1) model with zero mean in the presence of a single additive outlier.

\subsection{Robust estimation based on order statistics}

Several robust estimators of $\rho(h)$ are available in the literature, although they are usually computationally very demanding; see D\"urre, Fried and Liboschik (2015) for a survey on robust estimators of the autocorrelations. The high computational cost of most robust estimators seriously limit their empirical applications in the analysis of time series. However, the popular estimator proposed by Ma and Genton (2000), which is based on order statistics and on the scale estimator proposed by Rousseeuw and Croux (1993), is relatively simple from a computational point of view, requiring no more than $O(TlogT)$ computer time and $O(T)$ storage.\endnote{It can be computed using the fast algorithm described in Croux and Rousseeuw (1992).} This estimator is given by 
\begin{equation}
\label{eq:ma}
\hat{\rho}_{MG}(h)=\frac{V^2(u+v) - V^2(u-v)}{V^2(u+v) + V^2(u-v)},
\end{equation}
where $u=(y_1,...,y_{T-h})^{\prime}$ and $v=(y_{h+1},...,y_T)^{\prime}$ are the $T-h$ vectors of the first and last observations, and $V(z)=\left[ |z_i-z_j|; i>j \right]_{(1/4)}$, i.e. all absolute differences in increasing order are sorted and its 1/4 quantile is chosen. Note that the estimator in (\ref{eq:ma}) is guaranteed to take values in the interval $[-1,1]$. Furthermore, it is important to remark that being based on ranks and scales, $\hat{\rho}_{MG}(h)$ is location-free and, therefore, it does not depend on whether the mean is assumed to be known or not. 

The theoretical variance of $\hat{\rho}_{MG}(h)$ does not have closed form and its asymptotic distribution is unknown, limiting the possibility of inference on the estimated autocorrelations; see, for example, the empirical illustration carried out by Ma and Genton (2000) in which they do not provide significance bands for the estimated autocorrelations.

\subsection{Hurwicz robust estimator of the autoregressive parameter}

Hurwicz (1950) considers the related problem of estimating the autoregressive parameter of a stationary AR(1) model with zero mean, i.e. $y_t$ is given by model (\ref{eq:process}) with $c_i= \phi^{i}$ and $\mu=0$, where $|\phi|<1$. Hurwicz (1950) propose estimating the autoregressive parameter, $\phi$, as follows
\begin{equation}
\label{eq:hur}
\hat{\phi}=med\left(x_2, x_3,...,x_T\right),
\end{equation}
where $med\left(x_2,...,x_T \right)$ denotes the sample median of $x_2,...,x_T$ with $x_t=\frac{y_{t}}{y_{t-1}}$. One important attractive of the Hurwicz estimator in (\ref{eq:hur}) is that it does not require numerical optimization making it computationally simple. Furthermore, Hurwicz (1950) conjectures that $\hat{\phi}$ should be median unbiased and, consequently, a robust estimator. This conjecture was then proved by Zieli\'nski (1999), not only for Gaussian zero mean independent innovations but whenever the medians of independent (not necessarily identically distributed) innovations are equal to zero. 
 Later, Luger (2006) relax the condition of independence to serially uncorrelated innovations opening the door to dealing with AR(1) models with GARCH disturbances.

Berkoun, Fellag and Zieli\'nski (2003) analyse the size and power of the test of $H_0: \phi \neq 0$ in the presence of a single additive outlier. They show that the test is more powerful and with better size properties than when it is based on the LS estimator (sample autocorrelation in (\ref{eq:sample_1}) with $h=1$).

The asymptotic distribution of $\hat{\phi}$ is derived by Berkoun and Fellaf (2011) when the innovations are \textit{i.i.d.} symmetric variables in the domain of attraction of an $\alpha$-stable law ($1<\alpha<2$).\endnote{Note that infinite variance is allowed.} Denote by $F$ the common marginal distribution of $x_t$, for $t=1,...,T$, by $f$ their corresponding density, and by $q_{0.5}$ the $50\%$ quantile of $F$, i.e., the median. Assuming that: \textit{i)} $f$ is bounded in some neighbourhood, $V_0$, of $q_{0.5}$, with $0<q_{0.5}<\infty$ and $0<f(q_{0.5})<\infty$; \textit{ii)} $f'$ is bounded in $V_0$; \textit{iii)} $\inf_T \sigma^2_T>0$, where $\sigma^2_T=Var(\frac{S_T}{\sqrt{T}})$ with $S_T=\sum_{t=2}^T X_t$ and $X_t=0.5-I(x_t \leq q_{0.5})$, then the asymptotic distribution of $\hat{\phi}$ is given by
\begin{equation}
\label{eq:Asympt_H}
\frac{\sqrt{T} \left(\hat{\phi} - q_{0.5} \right) f(q_{0.5})}{\sigma_T} \xrightarrow{d} N(0,1)
\end{equation}
where $\sigma^2_T=E(X_1^2)+2 \sum_{h=1}^{T-1} \left(1-\frac{h}{T} \right) E \left(X_1 X_{1+h} \right)$.

Being based on the median, the estimator proposed by Hurwicz is not efficient. Consequently, pretending that the ratios $x_t$ are \textit{i.i.d.}, Reschenhofer (2019) proposes estimating $\phi=\rho(1)$ by a QML estimator based on the Cauchy distribution of $x_t$ with the following density\endnote{Note that Reschenhofer (2019) assumes that $y_t$ is Gaussian, although as we will see latter, this assumption can be relaxed to $y_t$ being distributed according to a elliptically distributed Pearson-type II, Pearson-type VII or Kotz-type distribution.}
\begin{equation}
\label{eq:Cauchy}
f(x;\rho(1))= \frac{\sqrt{1-\rho(1)^2}}{\pi \left( x^2-2 \rho(1) x +1\right)},
\end{equation}
which only depends on $\rho(1)$, the mode of the distribution.

Taking into account that the mean of a Cauchy variable is not defined, the estimator in (\ref{eq:hur}) is an obvious estimator of $\rho(h)$. However, it is well known that the median is not an efficient estimator. In particular, in the context of estimating $\rho(1)$ in an AR(1) model, Reschenhofer (2019) argues that the sample median has low Asymptotic Relative Efficiency, even when $\rho(1)$ is small (81\%) and decreases as $\rho(1)$ increases. Furthermore, from (\ref{eq:Cauchy}), it is easy to see that the first derivative of the log-density is the same for $x_t$ and $x_t^{-1}$, implying that the QML estimator remains unchanged when each observation $x_t$ such that $|x_t|>1$ is replaced by its inverse $x_t^{-1}$. Consequently, Reschenhofer (2019) proposes considering the following transformation of the ratios between successive observations
\begin{equation}
\label{eq:ratios_1}
z_t= sign\left(y_t y_{t-1}\right)\frac{\min \left( |y_t|, |y_{t-1}| \right)}{\max \left( |y_t|, |y_{t-1}| \right)},
\end{equation}
which by definition is restricted to take values in the interval $[-1,1]$. The distribution of $z_t$ is then given by the truncated Cauchy distribution proposed by Nadarajah and Kotz (2006) with density given by two times the Cauchy density in (\ref{eq:Cauchy}) in the interval  $[-1,1]$.\endnote{The density of the truncated-Cauchy distribution is given by
\begin{equation*}
f(z;\rho)=\frac{1}{D} \frac{\theta}{(z-\rho)^2 + \theta^2},
\end{equation*}
where $\theta=\sqrt{1-\rho^2}$ and $D=\arctan(\beta)-\arctan(\alpha)$ with $\alpha=\frac{A-\rho}{\theta}$ and $\beta=\frac{B-\rho}{\theta}$, and $[A,B]$ being the truncation interval, i.e. $-\infty < A \leq x \leq B < \infty$.  This distribution has well defined moments; see Nadarajah and Kotz (2006). In particular, the density of $z_t$, which is obtained when $A=-1$ and $B=1$, and, consequently, $\alpha=\frac{-(1+\rho)}{\theta}$ and $\beta=\frac{1-\rho}{\theta}$ with $-\alpha=\beta^{-1}$ and $D=\frac{\pi}{2}$, is given by
\begin{equation*}
f(z;\rho)=\frac{2}{\pi}  \frac{\theta}{(z-\rho)^2 + \theta^2},
\end{equation*}
with $E(z)=\rho+\frac{1-\rho^2}{\pi} \log\left( \frac{1-\rho}{1+\rho} \right)$. Furthermore, if $y_t$ is an uncorrelated sequence, then $\rho=0$ and, consequently,
\begin{equation*}
f(z;\rho)=\frac{2}{\pi}  \frac{1}{z^2 + 1}.
\end{equation*}
This latter distribution, known as truncated standard Cauchy, originally appeared in Johnson and Kotz (1970) and Rohatgi (1976), who derived its first two order moments; see also Dahiya, Staneski and Chaganthy (2001). In particular, $E(z)=0$ and $Var(z)=\frac{4-\pi}{\pi}$.} The QML estimator of $\rho(1)$ based on the truncated ratios, $z_t$, is the same as that based on the original ratios, $x_t$. Furthermore, since it is defined over a finite interval, the mean of the truncated Cauchy distribution is defined and given by
\begin{equation}
\label{eq:mean_cauchy}
E(z)=\rho(1)+ \frac{\sqrt{1-\rho(1)^2}}{\pi} \log \left( \frac{1-\rho(1)}{1+\rho(1)}\right) =g(\rho(1)).
\end{equation}
Consequently, Reschenhofer (2019) proposes estimating $\rho(1)$ by inverting the function $g(\rho(1))$ evaluated at $\bar{z}=\sum_{t=2}^T z_t$, the sample mean of $z_t$.  To simplify the inversion of the function $g(\bar{z})$, he proposes using the following approximation
\begin{equation}
\label{eq:Reschenhofer}
\hat{\rho}_R(1)=-1+2 \Phi_{0.295}(\bar{z}),
\end{equation}
where $\Phi_q(\cdot)$ is the cumulative probability function of the standard normal distribution evaluated at the quantil $q$; see Johnson and Kotz (1970) for a brief discussion of the realtive efficiency of the sample median and the sample mean as estimators of $\rho(1)$.

\section{Robust estimators of the autocorrelation function via forward ratios}
\label{sec:ratios}

In this section, we propose estimating the autocorrelation function of stationary linear models with independent innovations by extending the estimators of the first order autocorrelation of AR(1) processes.originally proposed by Hurwicz (1950) and modified by Reschenhofer (2019). We also derive the asymptotic distribution of the proposed estimators.

\subsection{The median-based estimator}

Inspired by the estimator proposed by Hurwicz (1950), we propose estimating the autocorrelation of order $h$ of $y_t$ represented by a general stationary process as in (\ref{eq:process}) with $\mu=0$, $\rho(h)$, as follows:
\begin{equation}
\label{eq:new_1}
\hat{\rho}^{*}_H(h)=med\left(\frac{y_{h+1}}{y_1}, \frac{y_{h+2}}{y_2},...,\frac{y_T}{y_{T-h}} \right).
\end{equation}

\begin{lemma}
\label{lem:unbiased}
Consider the stationary linear model in (\ref{eq:process}) with zero median innovations. Then, $\hat{\rho}^{*}_H(h)$ is a median unbiased estimator of $\rho(h)$, that is, $\rho(h)$ is the median of the distribution of $\hat{\rho}^{*}_H(h)$ as follows
\begin{equation}
\label{eq:munbiased}
Prob\left[\hat{\rho}^{*}_H(h)< \rho(h)|\rho(h) \right]=Prob\left[ \hat{\rho}^{*}_H(h)> \rho(h)|\rho(h)\right],
\end{equation}

\end{lemma}
\begin{proof}
The innovations of model (\ref{eq:process}) are white noise and, consequently, they are uncorrelated and zero mean. The lemma follows directly from the results in Zieli\'nski (1999) and Luger (2006) if the innovations are further zero median.
\end{proof}

The asymptotic distribution of the Hurwicz-based estimator, $\hat{\rho}^{*}_H(h)$, is derived in the following theorem when when $y_t$ is uncorrelated, i.e. $\rho(h)=0, \forall h$. 
\begin{theorem}
\label{theo:asympt_H}
If $\left\lbrace  y_t \right\rbrace_{t=1}^T$ is given by (\ref{eq:process}) with $\mu=0$, $c_i=0$, for $i=1, 2,...$, and the innovations being elliptically distributed Pearson-type II, Pearson-type VII or Kotz-type distribution, then the asymptotic distribution of the Hurwicz estimator of $\rho(h)$ is given by 
\begin{equation}
\label{eq:Asympt_H_2}
\sqrt{T} \hat{\rho}^*_H(h) \xrightarrow{d} N\left( 0,\frac{\pi^2}{4} \right).
\end{equation}
\end{theorem}
\begin{proof}
According to Nadarajah (2006), if $y_t$ is autocorrelated with zero mean, constant variance, and distributed according to a elliptically distributed Pearson-type II, Pearson-type VII or Kotz-type distribution, then the ratios $x_t=\frac{y_t}{y_{t-h}}$ have a Cauchy distribution with median equal to $\rho(h)$ and density given by\endnote{This result is satisfied regardless of whether the innovations in (\ref{eq:process}) are independent or just uncorrelated. Although the extension to conditionally heteroscedastic innovations is trivial, in this paper, we focus the discussion on independent innovations and left the extension to uncorrelated innovations for further research.}
\begin{equation}
\label{eq:Cauchy_2}
f(x;\rho(h))= \frac{\sqrt{1-\rho(h)^2}}{\pi \left( x^2-2 \rho(h) x +1\right)}.
\end{equation}
Using the asymptotic distribution derived by Berkoun and Fellaf (2011) in (\ref{eq:Asympt_H}) and the fact that the ratios $x_t$ are Cauchy, we can see that $f(q_{0.5})=\frac{1}{\pi \sqrt{1-\rho(h)^2}}$. Furthermore, by definition, $X_1$ has a binomial distribution taking values -0.5 and 0.5 each with probability 0.5, and, consequently, $E(X_1^2)=1/4$. Finally, given that $\rho(h)=0$, $E(X_1 X_{1+k})=0$ and, consequently, in this case, $\sigma_T^2=1/4$. Substituting adequately these quantities in (\ref{eq:Asympt_H}), the asymptotic distribution of $\hat{\rho}^*_H(h)$ is as given in (\ref{eq:Asympt_H_2}).
\end{proof}

Three comments about Theorem \ref{theo:asympt_H}. First, taking into account that the mean of a Cauchy variable is not defined, using the median as in (\ref{eq:new_1}) is an obvious estimator of $\rho(h)$. Second, using the asymptotic distribution in (\ref{eq:Asympt_H_2}), we can see that the Hurwicz estimator of $\rho(h)$ is less efficient than the corresponding sample autocorrelation, $r^*(h)$, supporting the well known result about the median not being an efficient estimator. The minimum relative efficiency loss (when $\rho(h)=0$) is given by $\pi^2/4=2.47$. The variance of the Hurwicz estimator is at least 2.5 times larger than the variance of the sample autocorrelation.\endnote{Note that the estimator proposed by Haddad (2000) is also based on the median and, consequently, expected to have the same issues associated with efficiency of the Hurwicz estimator. Haddad (2000) shows that his estimator is asymptotically robust. } Third, the asymptotic distribution obtained in Theorem \ref{theo:asympt_H} allows the construction of 95\% significance bounds for $\rho(h)$ as follows:
\begin{equation}
\label{eq:bounds_H}
\pm \frac{1.96 \times \pi}{2 \sqrt{T}}.
\end{equation}
Obviously, the significance bounds in (\ref{eq:bounds_H}) are wider than those of the corresponding sample autocorrelations.


In the more realistic case of $\mu$ being unknown, one can always implement the estimator in (\ref{eq:new_1}) with the observations centered with respect to the mean as follows
\begin{equation}
\label{eq:new_2}
\hat{\rho}_H(h)=med\left(\frac{y_{h+1}-\bar{y}}{y_1-\bar{y}}, \frac{y_{h+2}-\bar{y}}{y_2-\bar{y}},...,\frac{y_T-\bar{y}}{y_{T-h}-\bar{y}} \right);
\end{equation}
see Smadi, Jaber and Al'Zubi (2014), who consider this estimator in the context of an AR(1) model, and carry out limited Monte Carlo simulations to investigate its distribution. It is important to remark that the asymptotic distribution of $\hat{\rho}_H(h)$ is the same as that of $\hat{\rho}^*_H(h)$. To see that, note that $\bar{y}$, the sample mean is a consistent estimator of the population mean and, therefore, it converges in probability to $\mu$, which by the Slutski Theorem implies that $y_{t}-\bar{y}$ converges in distribution to $y_{t}-\mu$, which is elliptically Pearson-type II, Pearson-type VII or Kotz-type distributed. Furthermore, using the Mann-Wald Continuous Mapping Theorem (Mann and Wald, 1943), $\frac{y_{t+h}-\bar{y}}{y_t-\bar{y}}$ converges in distribution to $\frac{y_{t+h}-\mu}{y_t-\mu}$, which according to Nadarajah (2006) has a Cauchy distribution. Note that the moment of time corresponding to the median of the centered in (\ref{eq:new_2}) is the same as that corresponding to the median of the uncentered ratios in (\ref{eq:new}). Denote this moment of time as $\tau$. Then, $\hat{\rho}_{H}^*(h)=\frac{y_{\tau+h}}{y_{\tau}}$, while $\hat{\rho}_H(h)=\frac{y_{\tau+h}-\bar{y}}{y_{\tau}-\bar{y}}$. Therefore, the asymptotic distribution of  $\hat{\rho}_{H}^*(h)$ and $\hat{\rho}_{H}(h)$ coincide.

\subsection{Quasi-Maximum likelihood estimator}


Looking for more efficient estimators of the autocorrelation function and inspired by Reschenhofer (2019), we also propose a QML estimator of $\rho(h)$ based on the truncated Cauchy distribution of the transformed ratios $z_t$ defined in (\ref{eq:ratios_1}), which by definition are restricted to take values in the interval $[-1,1]$. Recall that the derivative of the density remains unchanged when $x$ is replaced by $x^{-1}$ and, consequently, $\frac{y_t}{y_{t-h}}$ and $\frac{y_{t-h}}{y_t}$ have the same distribution. If $y_t$ has a elliptically distributed Pearson-type II, Pearson-type VII and Kotz-type distribution, then $z_t$ has a truncated-Cauchy distribution as defined in Endnote 8. Consequently, $\rho(h)$ can be estimated by QML maximizing the corresponding log-likelihood obtained as if $z_t$ were serially independent. Denote this estimator as $\hat{\rho}^*_{Q}(h)$, which is given by
\begin{equation}
\label{eq:QML}
\begin{split}
\hat{\rho}^*_{Q}(h) & =\arg \max_{\rho(h)} \log \mathcal{L}(\rho(h);z_1,...,z_T) \\
& = \arg \max_{\rho(h)} \sum_{t=1}^T \left(- \log\left(\frac{\pi}{2} \right) + \frac{1}{2} \log (1-\rho(h)^2)-\log ((z_t-\rho(h))^2 + (1-\rho(h)^2)) \right)
\end{split}
\end{equation}
Given that, even if $y_t$ were independent, $z_t$ are not independent, we also define the alternative ML estimator, which is based on defining the likelihood function using only the observations of $z_t$, which are $h+1$ periods apart, i.e. $z_t^*=z_t, t=1, 1+(h+1), 1+2(h+1), 1+3(h+1),...$, and define the corresponding log-likelihood as follows:
\begin{equation}
\label{eq:QML_2}
\begin{split}
\hat{\rho}^{**}_{ML}(h) & =\arg \max_{\rho(h)} \log \mathcal{L}(\rho(h);z^*_1,...,z^*_{T^*}) \\
& = \arg \max_{\rho(h)} \sum_{i=1}^{T^*} \left(- \log\left(\frac{\pi}{2} \right) + \frac{1}{2} \log (1-\rho(h)^2)-\log ((z^*_{1+i(h+1)}-\rho(h))^2 + (1-\rho(h)^2)) \right).
\end{split}
\end{equation}
where $T^*=\left[ \frac{T}{h+1} \right]+1$.

Note that under standard conditions, the ML estimator is consistent. Furthermore, Dahiya, Staneski and Chaganthy (2001) give the following sufficient condition for the existence of a finite ML estimate of $\rho(h)$, namely $\frac{\sum(z_i-\rho(h))^2}{T}<\rho(h)^2+\frac{1}{3}$. This condition is necessary for sufficiently large values of $T$. Furthermore, they discuss that even for modest sample sizes, the probatility that the ML estimate of $\rho(h)$ fails to exist is extremelly small. Next, we establish the asymptotic distribution of $\hat{\rho}^{**}_{ML}(h)$ when $\rho(h)=0$.
\begin{theorem}
\label{th:Asympt}
Let $y_t$, $t=1,...,T$, be defined as in (\ref{eq:process}) with known mean $\mu=0$, variance $\sigma_y^2 < \infty$, and $c_i=0$, for $i=1, 2,...$ so that the autocorrelations $\rho(h)=0$, for $h=1, 2, ...$, with independent innovations with elliptically distributed Pearson-type II, Pearson-type VII and Kotz-type distributions. The asymptotic distribution of the ML estimator of $\rho(h)$ is as follows
\begin{equation}
\label{eq:asymp}
\sqrt{T} \hat{\rho}^{**}_{ML}(h)\xrightarrow{d} N(0, 2).
\end{equation}.
\end{theorem}

\begin{proof}
By choosing observations of $z_t$ that are $h+1$ periods appart, these observations are independent if $\rho(h)=0$. Therefore, the estimator in (\ref{eq:QML_2}) is a proper ML estimator, which implies that
\begin{equation}
\label{eq:asymp_2}
\sqrt{T} \hat{\rho}^{**}_{ML}(h)\xrightarrow{d} N\left( 0, \mathcal{I}(\rho(h))\right),
\end{equation}
where $\mathcal{I}(\rho(h))$ is the information matrix evaluated at the true value of $\rho(h)=0$. Denote $\rho(h)=\rho$. The score evaluated at $\rho=0$ is given by
\begin{equation*}
\begin{split}
\ell (\rho)|_{\rho=0} &= \frac{d}{d \rho} \left(- \log\left(\frac{\pi}{2} \right) + \frac{1}{2} \log (1-\rho(h)^2)-\log ((z^*-\rho(h))^2 + (1-\rho(h)^2)) \right)|_{\rho=0}\\
& = \left( -\frac{\rho}{1-\rho^2} +  \frac{2z^*}{(z^*_{1+i(h+1)}-\rho)^2 + (1-\rho^2)} \right)|_{\rho=0} =  \left( \frac{2z^*}{z^{*2} + 1} \right).
\end{split}
\end{equation*}

The information matrix is given by
\begin{equation*}
\begin{split}
\mathcal{I}(\rho) = E_{\rho} \left[ \ell (\rho)|_{\rho=0} \right]^2 & = \int_{-1}^1 \left( \frac{2z^*}{z^{*2} + 1} \right)^2 \frac{2}{\pi (1+z^{*2})} dz^* \\
& = \frac{8}{\pi} \int_{-1}^1 \frac{z^{*2}}{(1+z^{*2})^3} dz^* = \frac{8}{\pi} \frac{\pi}{16} = \frac{1}{2}.
\end{split}
\end{equation*}
\end{proof}
Note that when $T \rightarrow \infty$, the estimators in (\ref{eq:QML}) and in (\ref{eq:QML_2}) are equivalent. Consequently, the asymptotic distribution of $\hat{\rho}_{ML}^*(h)$ in (\ref{eq:asymp}) is also the asymptotic distribution of $\hat{\rho}_Q^*(h)$. Furthermore, using the same arguments as above, the more reallistic QML estimator based on $z_t$ obtained with observations of $y_t$ centered using the sample mean, denoted as $\hat{\rho}_Q(h)$, is also as in (\ref{eq:asymp}). Finally, the robust estimator $\hat{\rho}_{Q}(h)$ is less efficient than the sample autocorrelation, $r(h)$, with the relative asymptotic efficiency being 2, while it is more efficient than the Hurwicz-based estimator, $\hat{\rho}_H(h)$, with the relative asymptotic efficiency being 0.8.

The asymptotic distribution in (\ref{eq:asymp}) can be used to obtain significance intervals for the autocorrelations. Pointwise 95\% significance intervals with confidence for the autocorrelation of order $h$ can be constructed as follows
\begin{equation}
\label{eq:CI}
\pm \frac{1.96}{\sqrt{T/2}}.
\end{equation}

\subsection{Plug-in estimator based on sample means}

Alternatively, to avoid the computational issues associated to the numerical optimization of the likelihood, one can estimate $\rho(h)$ by using estimators based on the mean of the truncated-Cauchy distribution, which if $y_t$ is zero mean is given by
\begin{equation}
\label{eq:mean}
E(z_t)=\rho(h)+\frac{\sqrt{1-\rho(h)^2}}{\pi} log \left( \frac{1-\rho(h)}{1+\rho(h)}\right)=g(\rho(h)).
\end{equation}

The autocorrelation of order $h$ of $y_t$ can be estimated by an approximation of the inverse function of $g(\rho(h))$. We consider the following approximation of $g^{-1}(\rho(h))$
\begin{equation}
\label{eq:new}
\hat{\rho}^*_{P}(h)= \tanh \left(2.745 \bar{z}+1.034 \bar{z}^3\right),
 \end{equation}
where $\bar{z}=\sum_{t=1}^T z_t$. The constants in (\ref{eq:new}) have been chosen to minimize the mean square error in the approximation across all values of $\rho(h)$, which is smaller than that obtained when the approximation proposed by Reschenhofer (2019) in the context of $\rho(1)$ is implemented. The estimator in (\ref{eq:new}) is very simple computationally as it does not require numerical optimization.

The asymptotic distribution of the plug-in estimator in (\ref{eq:new}) can be obtained when $\rho(h)=0$ as established in the following theorem.
\begin{theorem}
Let $y_t$, $t=1,...,T$, be defined as in (\ref{eq:process}) with known mean $\mu=0$, variance $\sigma_y^2 < \infty$, and $c_i=0$, for $i=1, 2,...$ so that the autocorrelations $\rho(h)=0$, for $h=1, 2, ...$, with independent innovations with elliptically distributed Pearson-type II, Pearson-type VII and Kotz-type distributions. The asymptotic distribution of the plug-in estimator of $\rho(h)$ is as follows
\begin{equation}
\label{eq:asymp_23}
\sqrt{T} \hat{\rho}^{*}_{P}(h)\xrightarrow{d} N(0, 2.06).
\end{equation}.
\end{theorem}

\begin{proof}
The transformed ratios $z_t$ have a truncated standard Cauchy distribution with mean zero and variance $\frac{4-\pi}{\pi}$. Furthermore, they are serially uncorrelated
\begin{equation*}
Cov(z_t, z_{t-h})= E(z_t z_{t-h}) = E\left( \frac{y_{t+h}}{y_t} \frac{y_t}{y_{t-h}} \right) =0.
\end{equation*}
Therefore, the asymptotic distribution of the sample mean $\bar{z}$ is given by
\begin{equation*}
\sqrt{T} \bar{z} \xrightarrow{d} N\left(0, \frac{4-\pi}{\pi}\right).
\end{equation*}
The asymptotic distribution of $\hat{\rho}^*_{P}(h)$ can be obtained using the Delta method as follows
\begin{equation*}
 \sqrt{T} \hat{\rho}^*_{P}(h) \xrightarrow{d} N\left(0,  \frac{4-\pi}{\pi}\left(\frac{d }{d \bar{z}} \hat{\rho}^*_{P}(h)|_{\bar{z}=0} \right)^2\right).
\end{equation*}
Now, the derivative needed to obtain the asymptotic variance is given by
\begin{equation*}
\frac{d}{d\bar{z}} \tanh( 2.745 \bar{z} + 1.034 \bar{z}^3 )|_{\bar{z}=0} = \left(1-\tanh^2(2.745 \bar{z} + 1.034 \bar{z}^3) \right) \left( 2.745 + 3 \times 1.034 \bar{z}^2\right)|_{\bar{z}=0} = 2.745
\end{equation*}
Therefore, the asymptotic variance of $\sqrt{T} \rho^*_{P}(h)$ is given by $\frac{4-\pi}{\pi} \times 2.745^2=2.06$
\end{proof}

Note that asymptotically the efficiency rate of the plug-in estimator is nearly the same as that of the ML estimator. Finally, using the same arguments as above, the more reallistic plug-in estimator based on $z_t$ obtained with observations of $y_t$ centered using the sample mean, denoted as $\hat{\rho}_P(h)$, is also as in (\ref{eq:asymp}).

\section{Simulation results}
\label{sec:simu}

In this section, we carry Monte Carlo experiments to analyse the finite sample properties of the estimators proposed in this paper, namely, $\hat{\rho}_{H}(h), \hat{\rho}_{Q}(h)$ and $\hat{\rho}_P(h)$,  and their resistance against outliers. We compare these properties with those of the standard sample autocorrelations and of the robust estimator $\hat{\rho}_{MG}(h)$. We also analyse the finite sample properties of tests for lack of autocorrelation.

\subsection{Finite sample performance of estimators of the autocorrelation function}

For $t=1,...,T$, with $T=50, 100$ and 200, $R=5000$ simulated time series are generated by the following two models:
\begin{enumerate}
\item AR(1)
\begin{equation}
\label{eq:sim_1}
y^*_t=\phi y^*_{t-1} + \sqrt{1-\phi^2} \varepsilon_t,
\end{equation}
where $\phi=0.5$ and 0.9.
\item MA(1)
\begin{equation}
\label{eq:sim_2}
y^*_t= \frac{\varepsilon_t}{\sqrt{1+\theta^2}} -\theta \frac{\varepsilon_{t-1}}{\sqrt{1+ \theta^2}},
\end{equation}
where $\theta=0.5$ and 0.9.
\end{enumerate}

In both cases, $\varepsilon_t$ is generated by either a standard normal distribution or a standardized Student-7 distribution. Note that, in both models, $E(y^*_t)=0$, $E(y^{*2}_t)=1$. Furthermore, in model (\ref{eq:sim_1}), $\rho(h)=\phi^h$, while, in model (\ref{eq:sim_2}), $\rho(1)=-0.4$ and -0.5 when $\theta=0.5$ and 0.9, respectively, and $\rho(h)=0$, for $h \geq 2$. Each simulated series is contaminated with isolated additive outliers of size $\delta=0, 3$ and 5 as follows
\begin{equation}
\label{eq:out}
y_t=y^*_{t} + sign(\upsilon_t) \frac{\delta}{1- \lambda L} B_t,
\end{equation}
where $B_t$ has a Bernoulli distribution with $Prob(B_t)=0.10$, $\lambda=0$ and 0.5, and $\upsilon_t$ has a standard normal distribution; see Chen and Liu (1993), Franses and Haldrup (1994) and Ma and Genton (2000) for similar designs. Note that the outlier is additive when $\lambda$ is 0. This implies that the variable exhibits a one-time shot effect that produces a once-and-for-all peak at time $t$. By contrast, if $\lambda=0.5$, then the outlier has a longer effect that gradually dies out. This type of outliers is often denoted as Temporary Change.  

For each series, we estimate $\rho(h)$, for $h=1$ and 5, using: i) the sample autocorrelation, $r(h)$; ii) the robust estimator proposed by Ma and Genton (2000), $\hat{\rho}_{MG}(h)$; the estimator based on medians, $\hat{\rho}_{H}(h)$; the QML estimator, $\hat{\rho}_{Q}(h)$; and the plug-in estimator, $\hat{\rho}_{P}(h)$. It is important to remark that, in order to control for potential finite sample biases due to estimation of the mean, we consider each estimator assuming first that the mean is known (in this particular case, as the population mean is zero, the observations are not centred), and second, assuming that it is unknown with the observations centred by the sample mean.\endnote{Reschenhofer (2019) provides a detailed description of the negative bias of the sample autocorrelation of order 1 when the sample mean is estimated and compares the sampling distribution of the estimators of $\rho(1)$ in (\ref{eq:hur}) with some bias-corrected estimators based on the sample autocorrelations in the context of AR-GARCH models. However, he does not consider the effects of the presence of outliers.} Recall that the robust estimator $\hat{\rho}_{MG}(h)$ is location free and, consequently, it does not depend on whether the mean needs to be estimated or not. All results are obtained using Matlab codes programmed by the first author.

Tables \ref{tab:sim_AR} and \ref{tab:sim_MA} report the Monte Carlo averages and Mean Squared Errors (MSEs) of the estimators of the first order autocorrelation, $h=1$, when the simulated series are generated by the AR(1) model in (\ref{eq:sim_1}) contaminated as in (\ref{eq:out}) with $\lambda=0$ and by the MA(1) model in (\ref{eq:sim_2}) contaminated with $\lambda=0.5$, respectively.

Consider the results reported in Table \ref{tab:sim_AR}. First, as it is well known, we can observe that, regardless of the estimator considered, and of whether the series are contaminated or not, estimating the mean causes biases when $T$ is not large. These biases can be severe when the persistence is large, i.e. $\phi=0.9$. Furthermore, the MSE also increases when the mean needs to be estimated although the larger part of this increase is due to the bias. As expected, when $T$ is large, the bias is negligible and the MSE is nearly the same regardless of whether the mean is estimated or not. Second, the effect of isolated outliers on the bias of the sample autocorrelations can be severe; see also the results by Reschenhofer (2019). Once more, the bias is larger for the more persistent case. More importantly, we can observe that the bias hardly decreases with the sample size. With respect to the MSE of the sample autocorrelations, we can observe that it increases only due to the bias caused by the outliers. Third, when there is no contamination by outliers, the behaviour in terms of bias of all robust estimators is very similar to that of the sample autocorrelations. The price to pay for using these robust estimators is in terms of MSE. The most efficient estimator is MG, which has MSE similar to those of the sample autocorrelations. However, recall that the asymptotic distribution of the MG estimator is unknown and, consequently, inference is not allowed. Alternatively, as stated in the theory, among the proposed estimators based on ratios of consecutive observations seprated by $h$-period of time, the most inefficient estimator is H, which is based on the median; see also the results of the simulations carried out by Sri Ranganath (2018), who analyse and compare the bias and MSEs of $r(1)$ and $\hat{\rho}_H(1)$ when the innovations have a non-contaminated GED distribution. The MSEs of the QML and plug-in estimators are very similar to each other, and lye between those of the MG and H estimators. Furthermore, we can observe that the difference between the MSEs of the MG and the latter estimators is negligible when either the autocorrelation or the sample size, $T$, are large.  Fourth, the estimators based on the ratio of consecutive observations are highly resistant to outliers; see also Guo (2000), who also considers estimators based on ranks, and double exponential and logistic distributions.  While the sample autocorrelations may be completely missleading to estimate population autocorrelations, the proposed estimators are resistant against outliers, even if those are large. Furthermore, if $T$ is large enough, we can observe that the robust estimators are unbiased.

Consider now the results reported in Table \ref{tab:sim_MA} for the contaminated MA(1) with negative autocorrelations. We can observe that the conclusions obtained from Table \ref{tab:sim_AR} are in concordance with those obtained from Table \ref{tab:sim_MA}.

All in all, the simulation evidence support the theoretical results and allow us to recomend using the proposed plug-in estimator, which is highly resistant against outliers with only a relatively minor price in terms of efficiency.

\begin{landscape}
\begin{table}[th]
\caption{Monte Carlo results for the estimators of the first order autocorrelation when the series, generated by an AR(1) model with parameter $\phi$, zero mean, and normal errors, are contaminated by isolated additive outliers of size $\delta$: $r$, $\hat{\rho}_{MG}$, $\hat{\rho}_H$, $\hat{\rho}_Q$ and $\hat{\rho}_P$. The estimators with * are obtained assuming that the population mean (zero) is known.}
\label{tab:sim_AR}
\begin{tabular}{lccccccccccccccccccccccccc}
\hline
& \multicolumn{6}{c}{$T=50$} & & \multicolumn{6}{c}{$T=100$} & & \multicolumn{6}{c}{$T=200$}\\
& \multicolumn{2}{c}{$\delta=0$} & \multicolumn{2}{c}{$\delta=3$}& \multicolumn{2}{c}{$\delta=5$} && \multicolumn{2}{c}{$\delta=0$} & \multicolumn{2}{c}{$\delta=3$} & \multicolumn{2}{c}{$\delta=5$} & & \multicolumn{2}{c}{$\delta=0$} & \multicolumn{2}{c}{$\delta=3$} & \multicolumn{2}{c}{$\delta=5$} \\
 & Mean & MSE & Mean & MSE & Mean & MSE &&  Mean & MSE & Mean & MSE & Mean & MSE && Mean & MSE & Mean & MSE & Mean & MSE \\
\hline
& \multicolumn{20}{c}{$\phi=0.5$}\\
 \hline
r & 0.447 & 0.037 & 0.308 & 0.084 & 0.208 & 0.138 && 0.473 & 0.016 & 0.327 &  0.054 & 0.217 &  0.108 && 0.487 &  0.008 & 0.336 & 0.039 & 0.220 & 0.092\\\vspace{0.2cm}
$r^* $ & 0.480 & 0.032 & 0.340 & 0.071 & 0.238 & 0.122 && 0.488 & 0.015 & 0.343 & 0.048 & 0.232 & 0.099 && 0.495 & 0.008 & 0.344 & 0.037 & 0.227 & 0.088\\ \vspace{0.2cm}
$\hat{\rho}_{MG}$ & 0.441 & 0.047 & 0.380 & 0.062 & 0.416 & 0.051 && 0.469 &  0.020 & 0.411 & 0.031 & 0.448 & 0.024 && 0.486 &  0.010 & 0.426 & 0.017 &  0.464 & 0.012\\
$\hat{\rho}_H$& 0.453 & 0.083 & 0.401 & 0.094 & 0.398 &  0.096 && 0.475 & 0.040 & 0.419 & 0.049 & 0.414 & 0.052 && 0.488 & 0.019 & 0.429 & 0.026 & 0.423 & 0.029\\\vspace{0.2cm}
$\hat{\rho}^*_H$	& 0.485 & 0.076 & 0.431 & 0.086 & 0.424 & 0.091 && 0.489 & 0.038 & 0.432 &  0.046 & 0.426 & 0.050 && 0.497 & 0.019 & 0.436 &  0.025 & 0.429 & 0.028\\
$\hat{\rho}_Q$ & 0.449 & 0.063 & 0.396 & 0.072 & 0.400 & 0.067 && 0.474 & 0.030 & 0.418 & 0.038 & 0.417 & 0.037 && 0.487 & 0.015 & 0.429 & 0.020 & 0.426 & 0.021\\\vspace{0.2cm}
$\hat{\rho}^*_Q$ & 0.479 &	0.058 & 0.424 & 0.064 & 0.420 & 0.063 && 0.489 & 0.029 & 0.431 & 0.035 & 0.428 & 0.035 && 0.495 & 0.014 & 0.436 & 0.019 & 0.431 & 0.020\\
$\hat{\rho}_P$ & 0.450	 & 0.066 & 0.404 & 0.074 & 0.418 &  0.067 && 0.475 & 0.031 & 0.427 & 0.039 & 0.436 & 0.036 && 0.488 & 0.016 & 0.438 & 0.020 &	 0.446 & 0.019\\\vspace{0.2cm}
$\hat{\rho}^*_P$ & 0.479 & 0.061 & 0.433 & 0.066 & 0.437 & 0.063 && 0.489 & 0.030 & 0.441 & 0.036 &	0.447 & 0.034 && 0.496 & 0.015 & 0.445 & 0.019 & 0.451 &	 0.018\\
\hline
& \multicolumn{20}{c}{$\phi=0.9$}\\
 \hline
r & 0.819 & 0.023 & 0.501 & 0.249 & 0.328 & 0.441 && 0.861 & 0.008 &	 0.551 & 0.166 & 0.350 & 0.356 && 0.881 & 0.003 & 0.588 & 0.121 & 0.374 & 0.303\\\vspace{0.2cm}
$r*$ & 0.870 & 0.013 & 0.594 & 0.173 & 0.417 & 0.346 && 0.884 & 0.005 & 0.601 & 0.130 & 0.401 & 0.304 && 0.891 & 0.002 & 0.614 & 0.104 & 0.401 & 0.276\\\vspace{0.2cm}
$\hat{\rho}_{MG}$ & 0.803 & 0.029 & 0.773 & 0.039 & 0.794 &  0.032 && 0.853 & 0.009 & 0.825 & 0.014 & 0.844 & 0.010 && 0.877 & 0.003 & 0.851 & 0.006 & 0.868 & 0.004\\
$\hat{\rho}_H$ & 0.835 & 0.040 & 0.786 & 0.063 & 0.788 & 0.060 && 0.869 & 0.014 & 0.827 & 0.024 & 0.825 & 0.024 && 0.884 & 0.006 & 0.846 & 0.011 & 0.845 &	 0.011\\\vspace{0.2cm}
$\hat{\rho}^*_H$ & 0.881 & 0.026 & 0.839 & 0.041 & 0.839 & 0.041 &&  0.892 & 0.011 & 0.852 & 0.018 & 0.850 & 0.018 && 0.894 & 0.005 & 0.859 & 0.009 & 0.857 & 0.009\\
$\hat{\rho}_Q$ & 0.825 & 0.030 & 0.774 & 0.052 & 0.778 & 0.048 && 0.865 & 0.010 & 0.819 & 0.021 & 0,817 & 0.021 && 0.883 & 0.004 & 0.842 & 0.010 & 0.839 & 0.010\\\vspace{0.2cm}
$\hat{\rho}^*_Q$ & 0.873 & 0.018 & 0.830 & 0.032 & 0.828 & 0.031 && 0.888 & 0.007 & 0.845 & 0.015 & 0.843 & 0.015 && 0.893 & 0.003 & 0.855 & 0.008 & 0.851 & 0.008\\
$\hat{\rho}_P$ & 0.827&  0.030 & 0.781 & 0.047 & 0.789 & 0.041 && 0.865 & 0.010 & 0.820 & 0.019 & 0.824 & 0.018 && 0.881 & 0.004 & 0.840 & 0.009 & 0.842 & 0.008\\\vspace{0.2cm}
$\hat{\rho}^*_P$ & 0.873 &	0.018 & 0.831 & 0.029 & 0.834 & 0.026 && 0.887 & 0.007 &  0.844 & 0.014 & 0.846 & 0.013 && 0.891 & 0.003 & 0.852 & 0.007 & 0.852 & 0.007\\
\hline
\end{tabular}
\end{table}
\end{landscape}

\begin{landscape}
\begin{table}[th]
\caption{Monte Carlo results for the estimators of the first order autocorrelation when the series, generated by an MA(1) model with parameter $\theta$, zero mean, and normal errors, are contaminated by additive outliers of size $\delta$: $r$, $\hat{\rho}_{MG}$, $\hat{\rho}_H$, $\hat{\rho}_Q$ and $\hat{\rho}_P$. The estimators with * are obtained assuming that the population mean (zero) is known.}
\label{tab:sim_MA}
\begin{tabular}{lccccccccccccccccccccccccc}
\hline
& \multicolumn{6}{c}{$T=50$} & & \multicolumn{6}{c}{$T=100$} & & \multicolumn{6}{c}{$T=200$}\\
& \multicolumn{2}{c}{$\delta=0$} & \multicolumn{2}{c}{$\delta=3$}& \multicolumn{2}{c}{$\delta=5$} && \multicolumn{2}{c}{$\delta=0$} & \multicolumn{2}{c}{$\delta=3$} & \multicolumn{2}{c}{$\delta=5$} & & \multicolumn{2}{c}{$\delta=0$} & \multicolumn{2}{c}{$\delta=3$} & \multicolumn{2}{c}{$\delta=5$} \\
 & Mean & MSE & Mean & MSE & Mean & MSE &&  Mean & MSE & Mean & MSE & Mean & MSE && Mean & MSE & Mean & MSE & Mean & MSE \\
\hline
& \multicolumn{20}{c}{$\theta=0.5$}\\
 \hline
r & -0.394 & 0.024 & -0.112 & 0.154 & 0.073 & 0.321 && -0.395 & 0.012 & -0.089 & 0.134 & 0.119 & 0.314 && -0.399 & 0.006 & -0.077 &	 0.123 & 0.140 & 0.314\\\vspace{0.2cm}
$r^*$ & -0.388 & 0.024 & -0.088 & 0.172 & 0.102 & 0.355 && -0.392 & 0.012 & -0.076 & 0.143 &	 0.134 & 0.330 && -0.398 & 0.006 & -0.070 & 0.128 & 0.148 & 0.323 \\\vspace{0.2cm}
$\hat{\rho}_{MG}$ & -0.397 & 0.034 & -0.231 & 0.095 & -0.190 & 0.124 && -0.396 & 0.017 & -0.226 & 0.063 & -0.191 & 0.084 && -0.401 & 0.008 & -0.227 & 0.047 &  -0.195 & 0.062 \\
$\hat{\rho}_H$ & -0.403 & 0.069 & -0.236 & 0.134 & -0.130 & 0.206 && -0.401 & 0.036 & -0.236 & 0.081 & -0.151 & 0.131 && -0.402 & 0.018 &	-0.239 & 0.053 &	-0.165 & 0.091\\\vspace{0.2cm}
$\hat{\rho}^*_H$	& -0.398 &	0.071 & -0.244 & 0.128 & -0.178 & 0.170 && -0.397 & 0.035 & -0.240 &	0.079 & -0.180 & 0.110 && -0.400 & 0.019 & -0.241 & 0.052 & -0.181 &	 0.081\\
$\hat{\rho}_Q$ & -0.397 & 0.051 & -0.227 & 0.114 & -0.122 & 0.189 && -0.396 & 0.027 &	-0.229 & 0.071 & -0.146 & 0.121 && -0.401 & 0.014 & -0.234 & 0.049 & -0.161 & 0.086\\\vspace{0.2cm}
$\hat{\rho}^*_Q$ & -0.391&	0.052 & -0.235 & 0.107 & -0.173 &	0.144 && -0.393 & 0.027 & -0.233 & 0.068 & -0.175 & 0.099 && -0.400 & 0.014 & -0.237 & 0.047 & -0.177 & 0.076\\
$\hat{\rho}_P$ &	-0.396 & 0.053 &	-0.226 & 0.115 &	-0.122 & 0.192 && -0.394 & 0.028 & -0.228 & 0.072 & -0.149 & 0.120 && -0.400 & 0.015 &-0.233 & 0.049 & -0.164 & 0.084\\\vspace{0.2cm}
$\hat{\rho}^*_P$ & -0.391 & 0.054 & -0.234 & 0.108 & -0.175 & 0.143 && -0.391 & 0.028 & -0.233 & 0.069 & -0.179 & 0.096 && -0.399 & 0.015 & -0.235 & 0.048 & -0.181 & 0.074\\
\hline
& \multicolumn{20}{c}{$\theta=0.9$}\\
 \hline
r & -0.488 & 0.020 & -0.173 & 0.177 & 0.033 & 0.387 && -0.492 & 0.010 & -0.147 & 0.161 & 0.078 & 0.381 && -0.495 & 0.005 & -0.138 &	 0.150 & 0.103 & 0.383 \\\vspace{0.2cm}
$r^*$ & -0.488 & 0.020 & -0.151 & 0.198 & 0.062 & 0.427 && -0.491 & 0.010 & -0.135 & 0.171 &	 0.094 & 0.401 && -0.495 & 0.005 & -0.131 & 0.155 & 0.111 & 0.393\\\vspace{0.2cm}
$\hat{\rho}_{MG}$ & -0.490 & 0.029 & -0.313 & 0.096 & -0.277 & 0.129 && -0.492 & 0.013 & -0.312 & 0.067 &  -0.284 & 0.086 && -0.496 &	0.007 & -0.317 &	0.049 & -0.289 &	0.063 \\
$\hat{\rho}_H$ & -0.493 & 0.064 & -0.315 & 0.137 & -0.209 & 0.220 && -0.492 & 0.032 & -0.322 & 0.081 & -0.232 & 0.142 && -0.497 & 0.015 & -0.332 & 0.053 & -0.256 & 0.094\\\vspace{0.2cm}
$\hat{\rho}^*_H$ & -0.492 &	0.064 & -0.325 &	0.128	& -0.272 & 0.170 && -0.492 & 0.031 & -0.333 & 0.076 &-0-272 &	0.113 && -0.497 & 0.015 & -0.337 & 0.051 & -0.278 & 0.080\\
$\hat{\rho}_Q$ &	-0.488 & 0.044 &	-0.303 & 0.116 & -0.198 & 0,204 && -0.490 & 0.022 & -0.316 & 0.072 & -0.224 & 0.133 && -0.495 & 0.011 & -0.324 & 0.051 & -0.247 & 0,091\\\vspace{0.2cm}
$\hat{\rho}^*_Q$ & -0.488 & 0.044 & -0.318 &	 0.107 & -0.263 &	0.148 && -0.490 & 0.022 & -0.325 & 0,069 & -0.264 & 0.103 && -0.495 & 0.011 & -0.329 & 0.049 & -0.268 & 0.077\\
$\hat{\rho}_P$ &	-0.488 & 0.045 & -0.304 & 0.117 & -0.199 &  0.206 && -0.490 & 0.023 & -0.315 & 0.073 & -0.227 & 0.131 && -0.495 & 0.011 & -0.322 & 0.052 & -0.249 & 0.089\\\vspace{0.2cm}
$\hat{\rho}^*_P$ & -0.488 & 0.045 & -0.318 & 0.088 & -0.265 & 0.145 && -0.490 & 0.023 & -0.324 & 0.069 & -0.266 &	 0.101 && -0.495 & 0.011 & -0.327 & 0.050 & -0.270 & 0.075\\
\hline
\end{tabular}
\end{table}
\end{landscape}

To further analyse the finite sample distribution of the proposed estimators of the autocorrelations, Figure \ref{fig:densities_1} plots the Monte Carlo densities of the five estimators of $\rho(5)$ considered when the series are simulated  by the AR(1) model with $\phi=0.9$, i.e. $\rho(5)=0.59$ and contaminated with outliers of size $\delta=0$ and 3, and with $\lambda=0$ and 0.5. Several important conclusions can be extracted from Table \ref{fig:densities_1}. First, we can observe that, when $T$ is small, the distribution of the sample autocorrelation when there are not outliers is extremely skewed to the left; see Reschenhofer (2019) for the same result in the context of the first order autocorrelation in AR(1) models with known mean. Similarly, the distributions of all robust estimators but the Ma and Genton (2000), are also left-skewed.  However, the distributions are more symmetric as the sample size increases. Second, the performance of $\hat{\rho}_{MG}(h)$ in the presence of outliers is similar to that of the sample autocorrelations. It seems that $\hat{\rho}_{MG}(h)$ looses part of its robustness properties when implemented to estimate autocorrelations of large order due to a reduction in the number of pairs used to estimate the autocorrelation. Also note that the performance of $\hat{\rho}_{MG}$ depends on the degree of persistence of the underlying time series. Third, the estimators of the autocorrelations based on ratios of consecutive observations separated by $h$ periods of time have large resistance to the presence of outliers. These distributions are hardly affected by the presence of outliers. Finally, it is remarkable that the empirical distributions of the proposed $\hat{\rho}_Q(h)$ and $\hat{\rho}_{P}(h)$ estimators are very similar to each other.

It is important to note that the empirical distribution of the robust estimators proposed in this paper approaches normality as $T$ increases. Even though their asymptotic distributions have only been derived for $\rho(h)=0$, we can conjecture that it is also valid for $\rho(h) \ne 0$.

\begin{figure}[h!]
\includegraphics[scale=0.24]{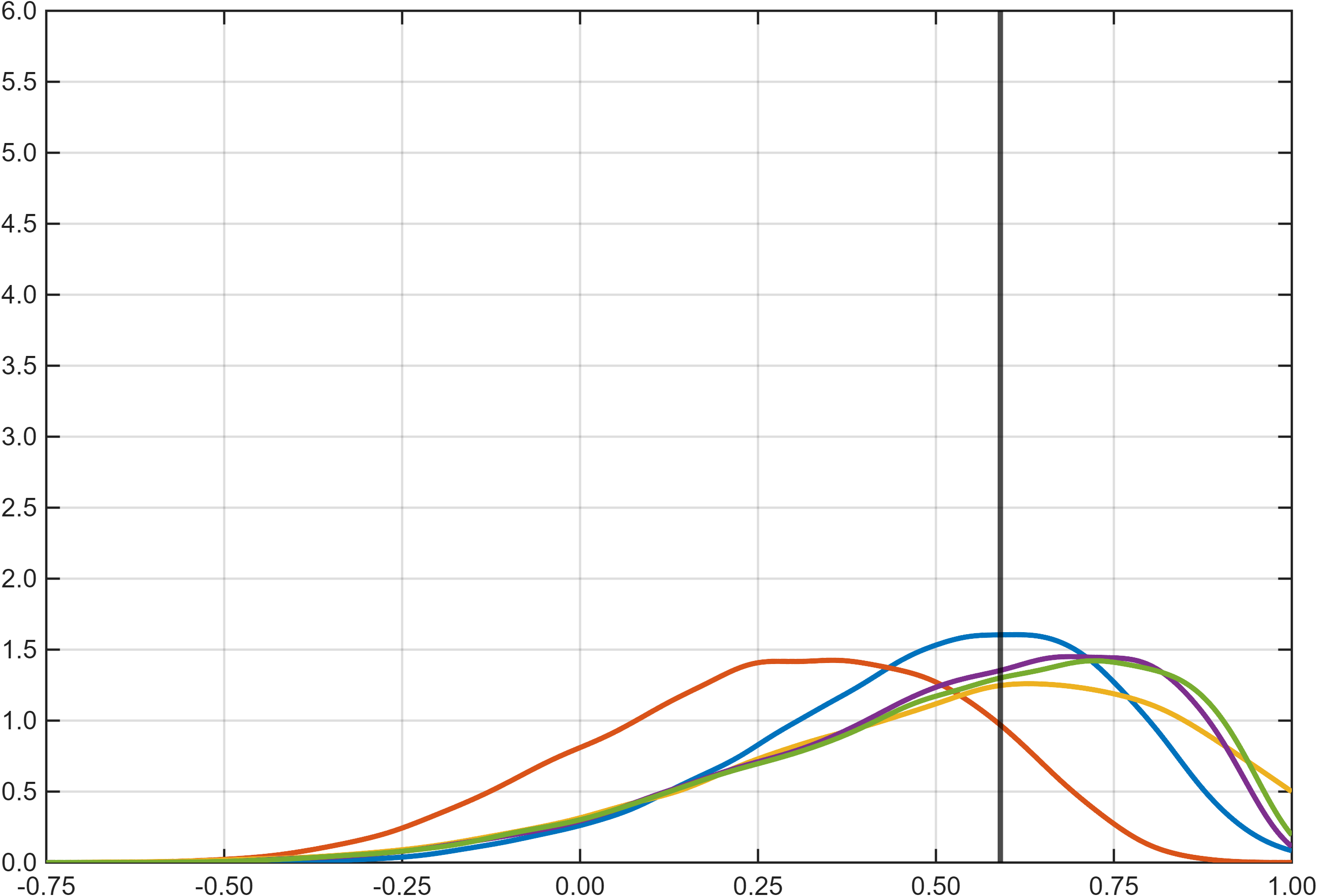}
\includegraphics[scale=0.24]{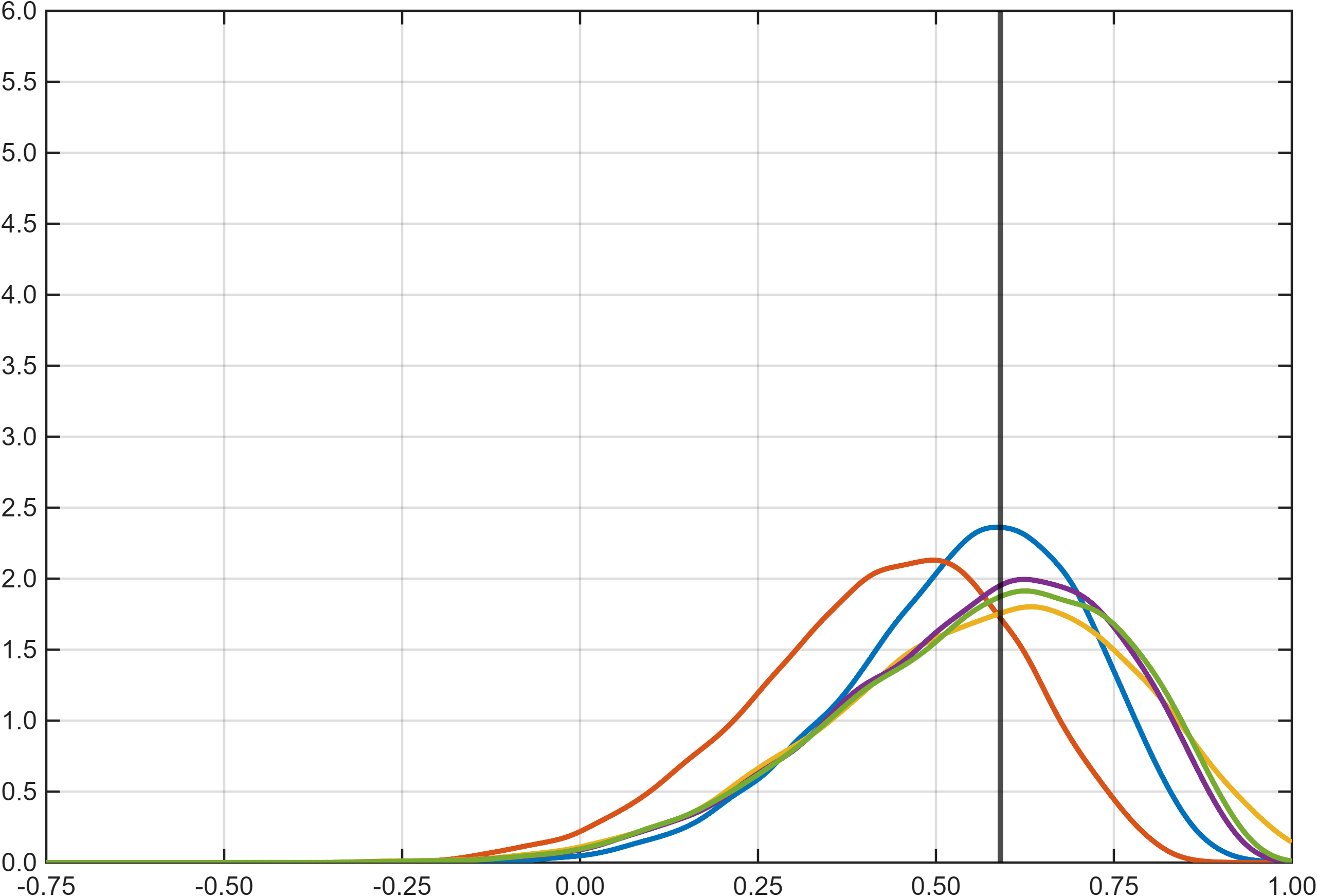}
\includegraphics[scale=0.23]{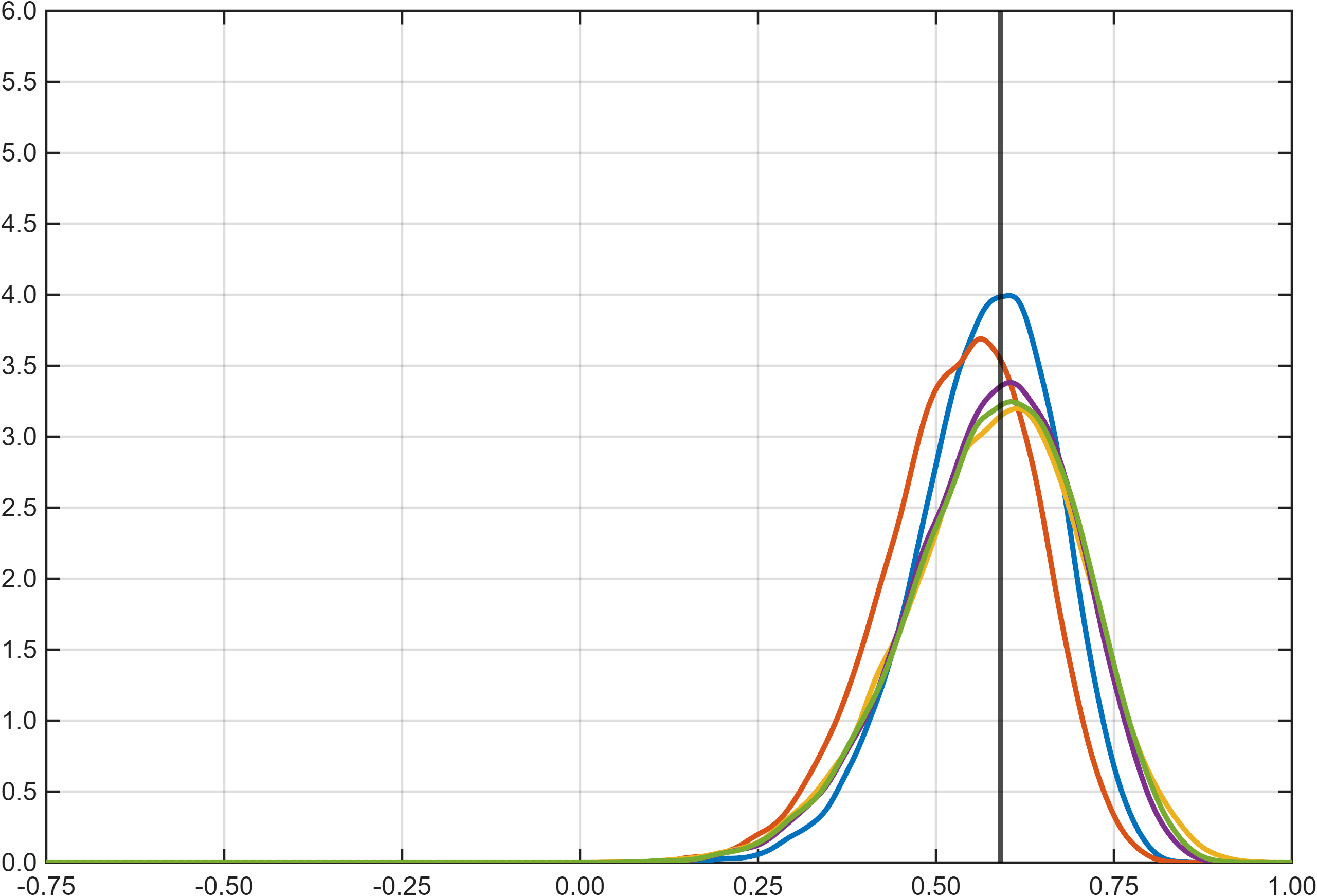}\\
\includegraphics[scale=0.24]{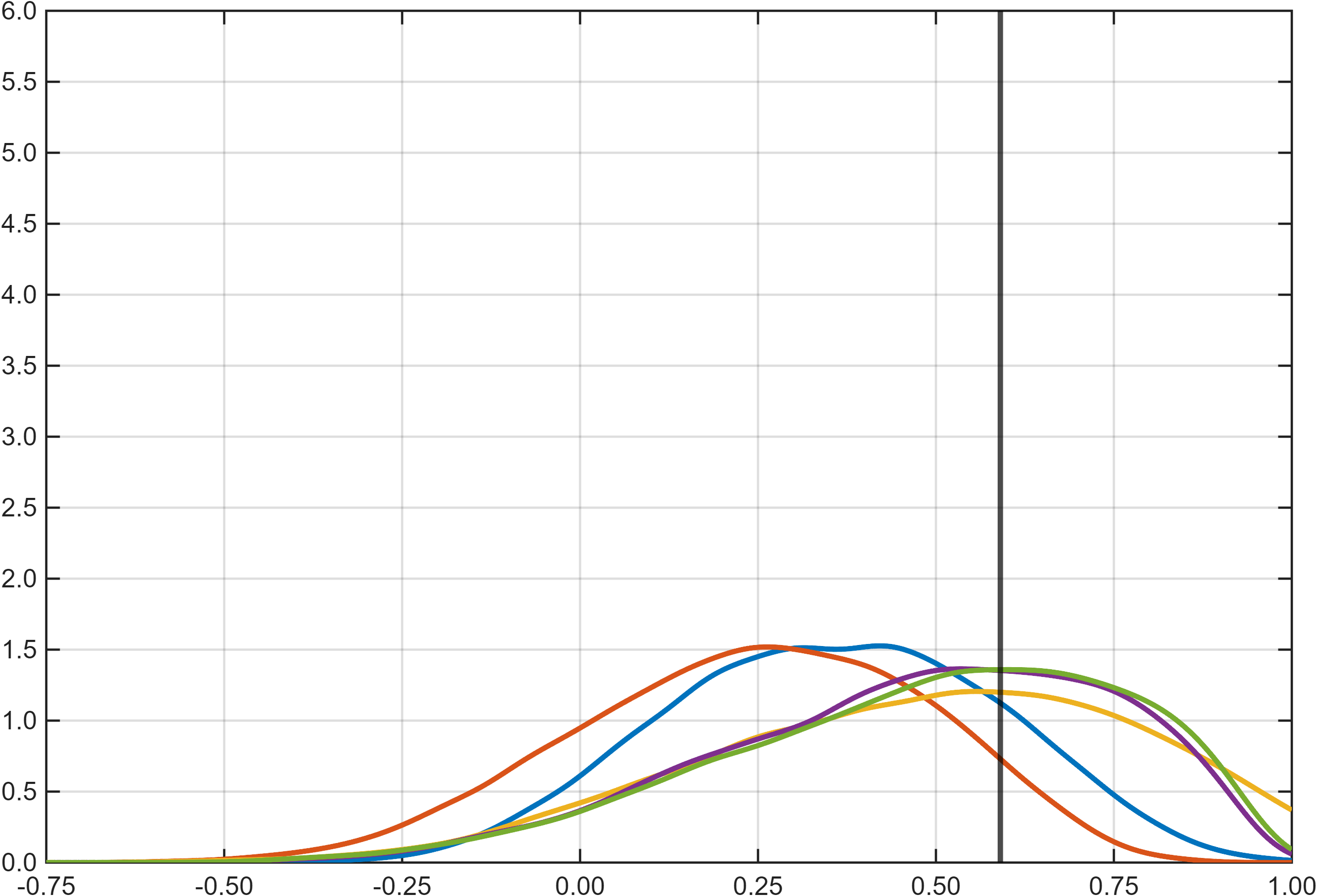}
\includegraphics[scale=0.24]{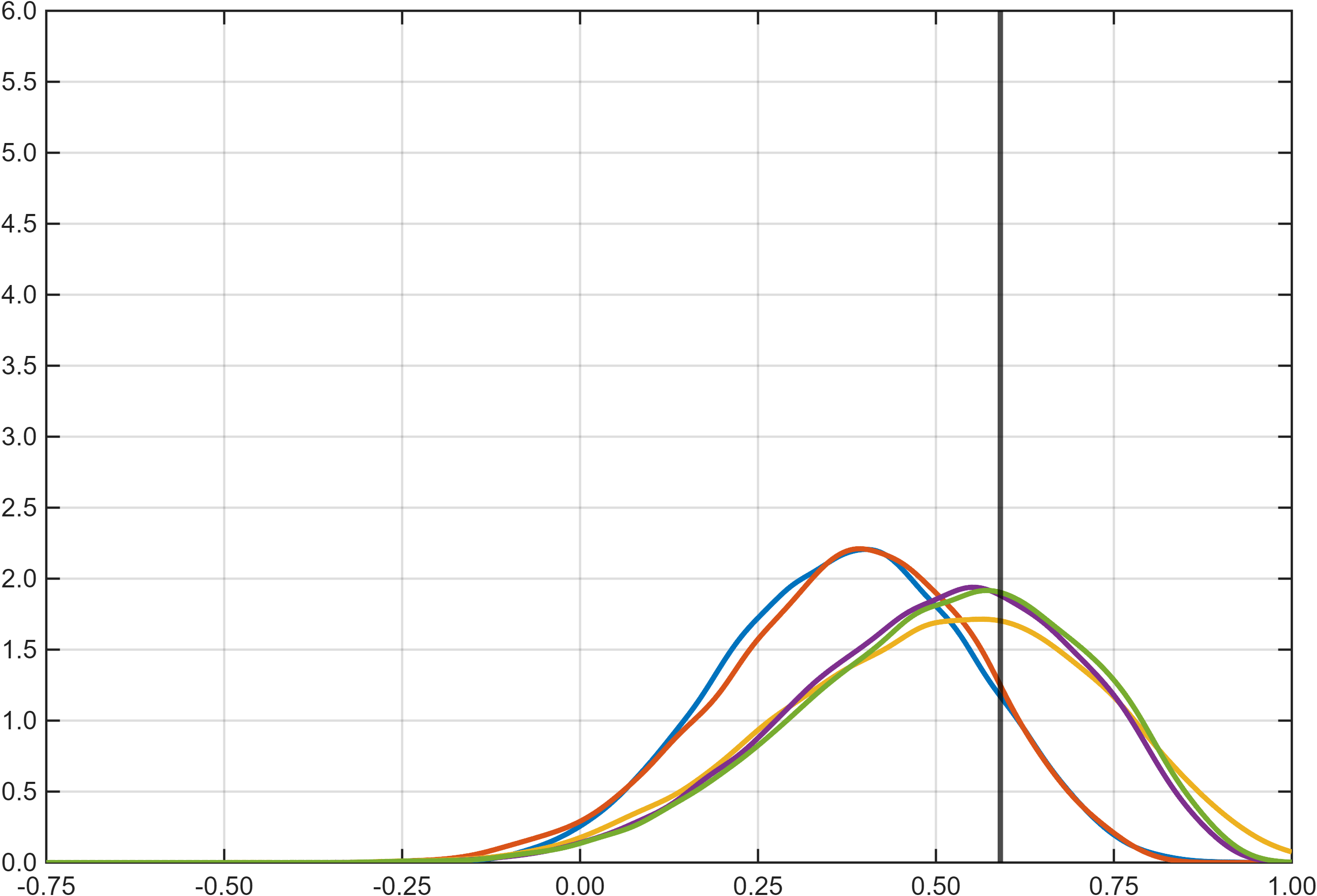}
\includegraphics[scale=0.24]{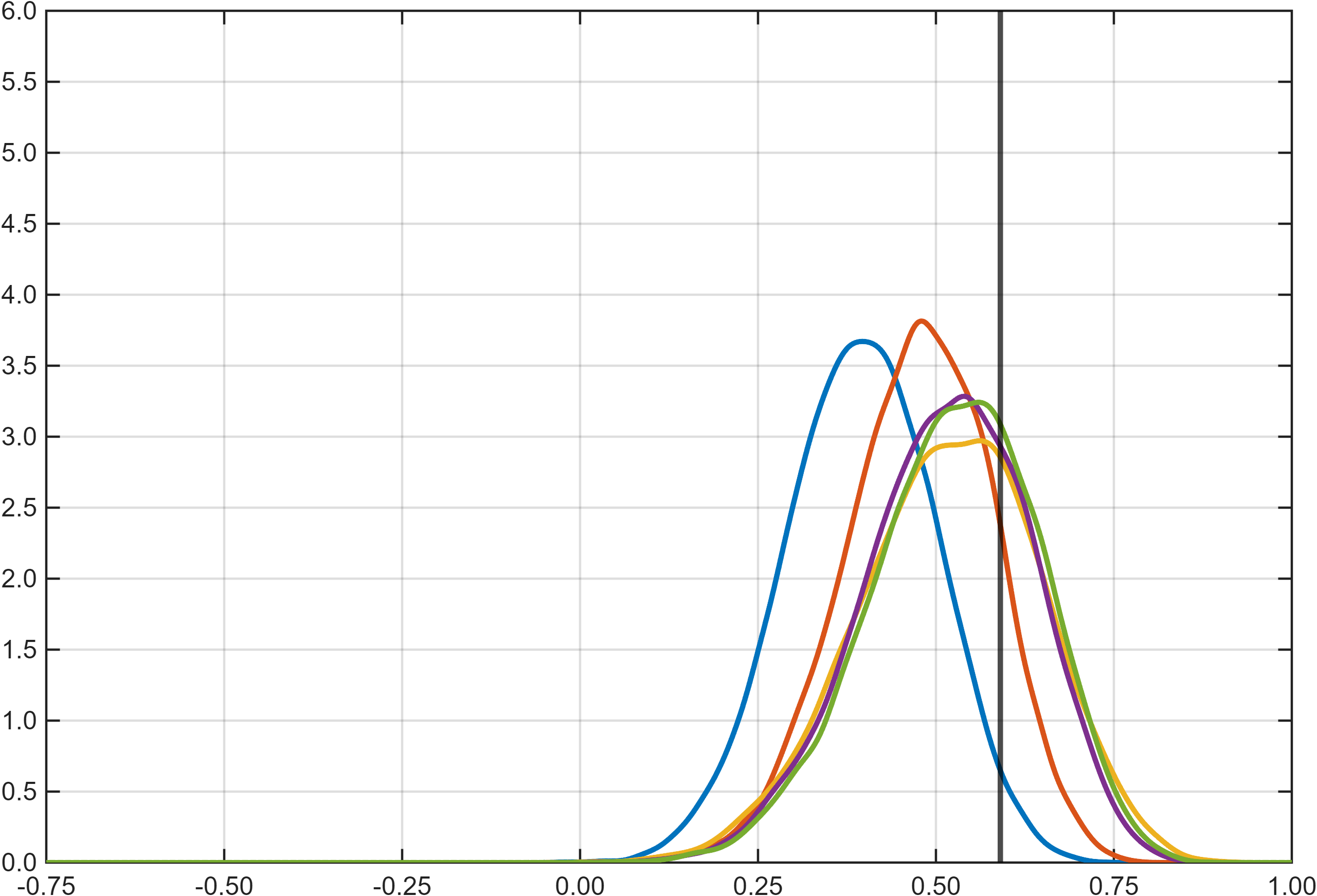}\\
\includegraphics[scale=0.24]{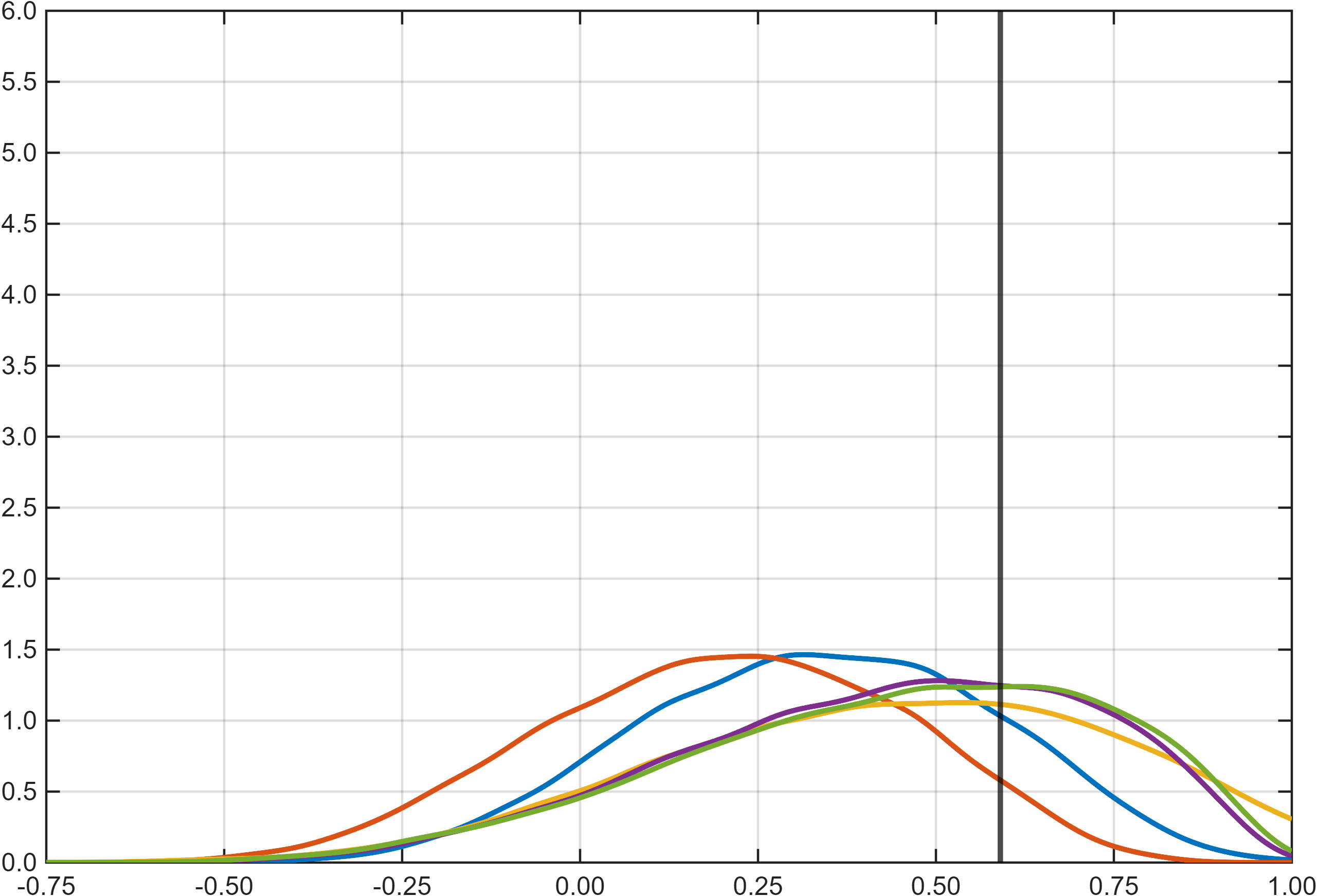}
\includegraphics[scale=0.24]{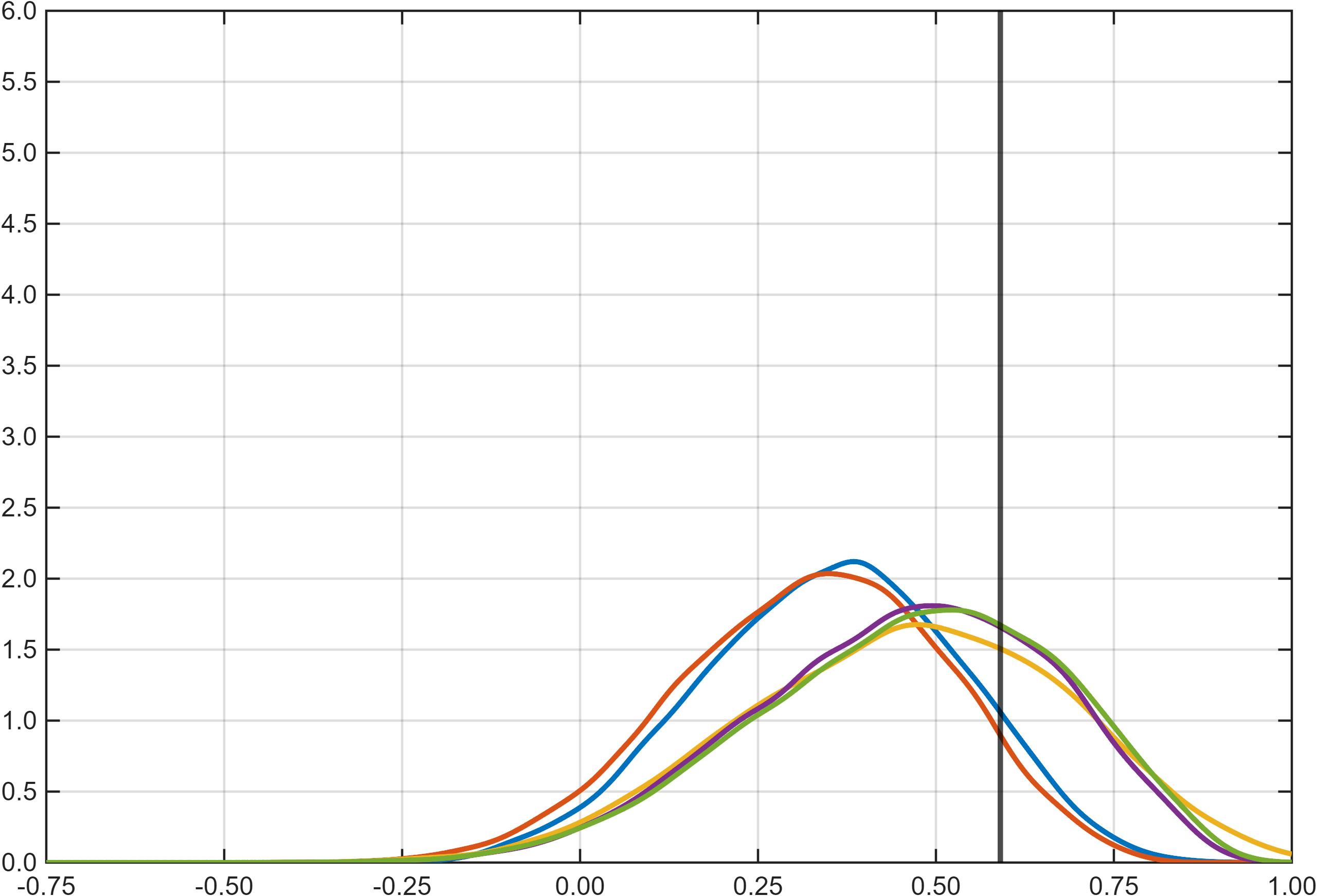}
\includegraphics[scale=0.24]{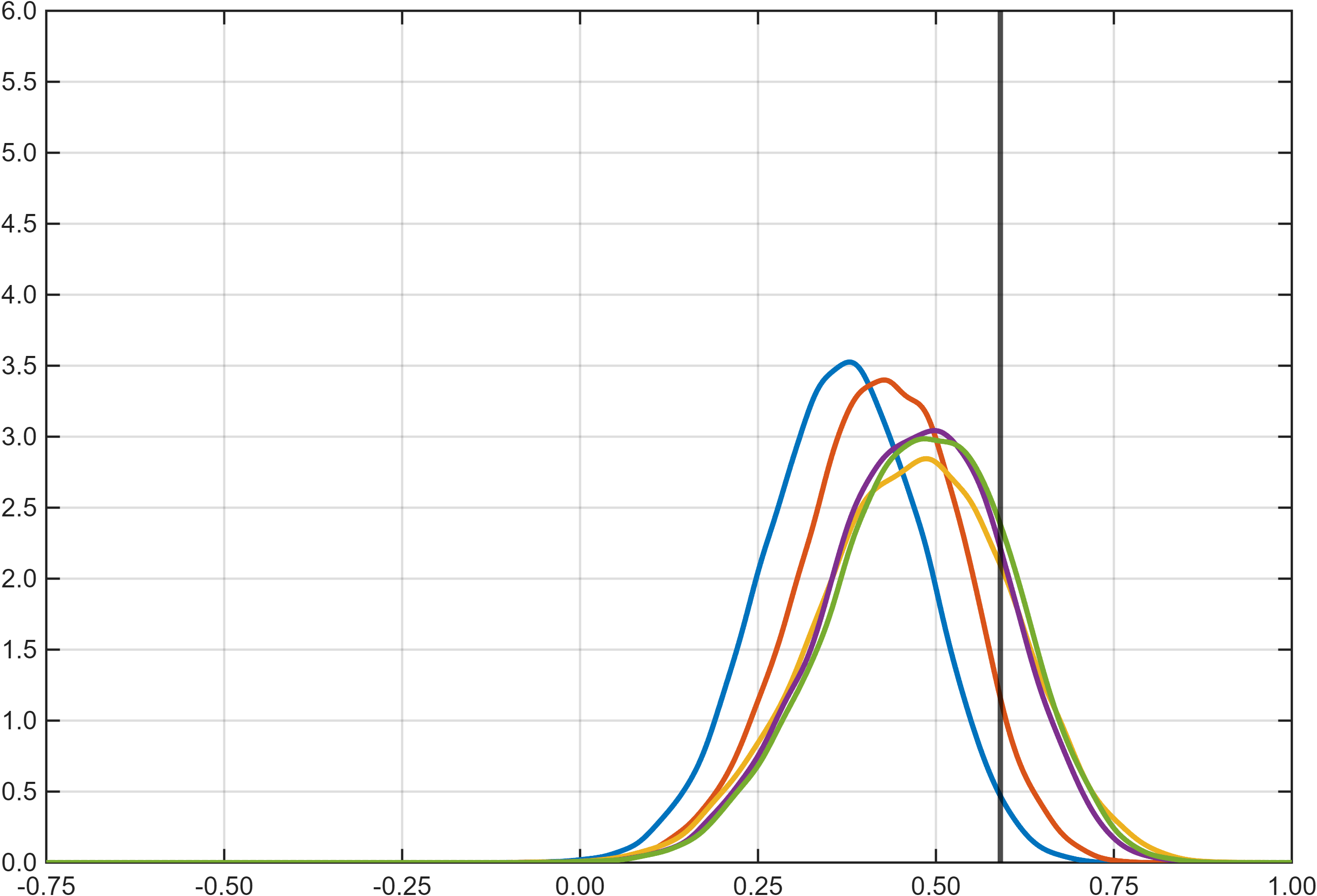}
\caption{Monte Carlo densities of estimators of $\rho(5)=0.59$ (vertical black lines) obtained when the simulated series of size $T=50$ (first column), $T=100$ (second column) and $T=300$ (third column) are generated with the AR(1) model with $\phi=0.9$ and contaminated with isolated additive outliers with $\delta=0$ (top row), $\delta=3$ and $\lambda=0$ (middle row), and $\delta=3$ and $\lambda=0.5$ (bottom row): $r(5)$ (blue), $\hat{\rho}_{MG}(5)$ (orange), $\hat{\rho}_H(5)$ (yellow), $\hat{\rho}_Q(5)$ (purple), and $\hat{\rho}_P(5)$ (green).}
\label{fig:densities_1}
\end{figure}

Finally, Figure \ref{fig:densities_2} plots the Monte Carlo densities of the estimators of the autocorrelation of order 5 when the series are simulated by the AR(1) model with parameter $\phi=0$, together with the finite sample approximation of the asymptotic densities of $r(5)$, $\hat{\rho}_H(5)$ and of $\hat{\rho}_{Q}(5)$ and $\hat{\rho}_P(5)$ in (\ref{eq:asymp_rho}), (\ref{eq:Asympt_H_2}) and (\ref{eq:asymp}), respectively. Several important conclusions can be obtained. First, we can observe that, regardless of the sample size, in absence of outliers the asymptotic distributions of the estimators of the sample autocorrelations considered  are good approximations of their corresponding empirical distributions. Therefore, the pointwise significance intervals constructed using the asymptotic distribution should have the appropriate coverage even in small samples. However, as pointed out above, the empirical distribution of the MG estimator is asymmetric and somehow between those of the sample autocorrelations and of the QML and plug-in estimators, which are very similar between them. 
Second, when the series are contaminated by outliers, the empirical distributions of the $\hat{\rho}_Q(5)$ and $\hat{\rho}_P(5)$ estimators are almost identical and close to their asymptotic distribution. The distribution of the sample autocorrelations are biased and not well approximated by the asymptotic distribution even when the sample size is large.
\begin{figure}[!ht]
\includegraphics[scale=0.24]{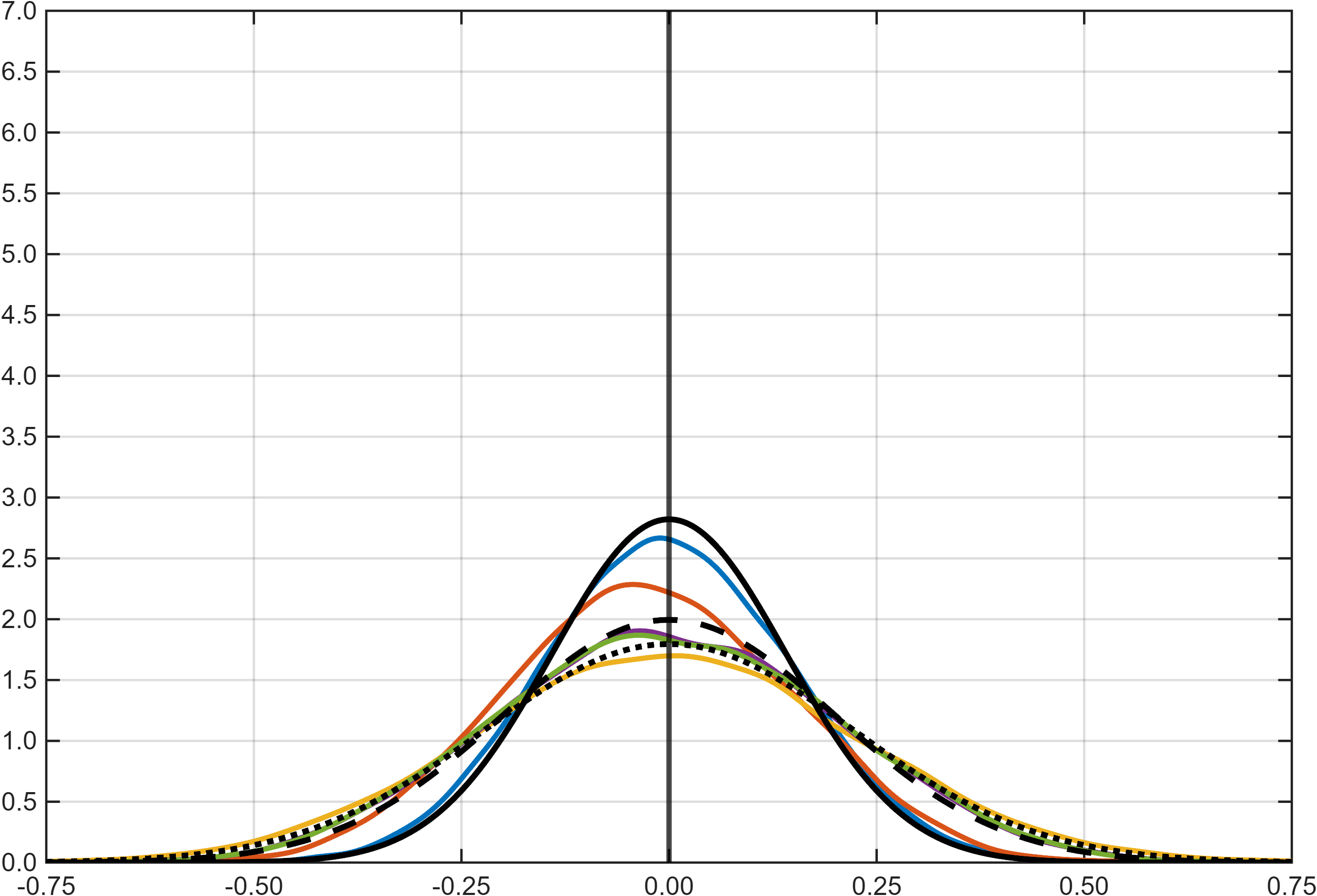}
\includegraphics[scale=0.24]{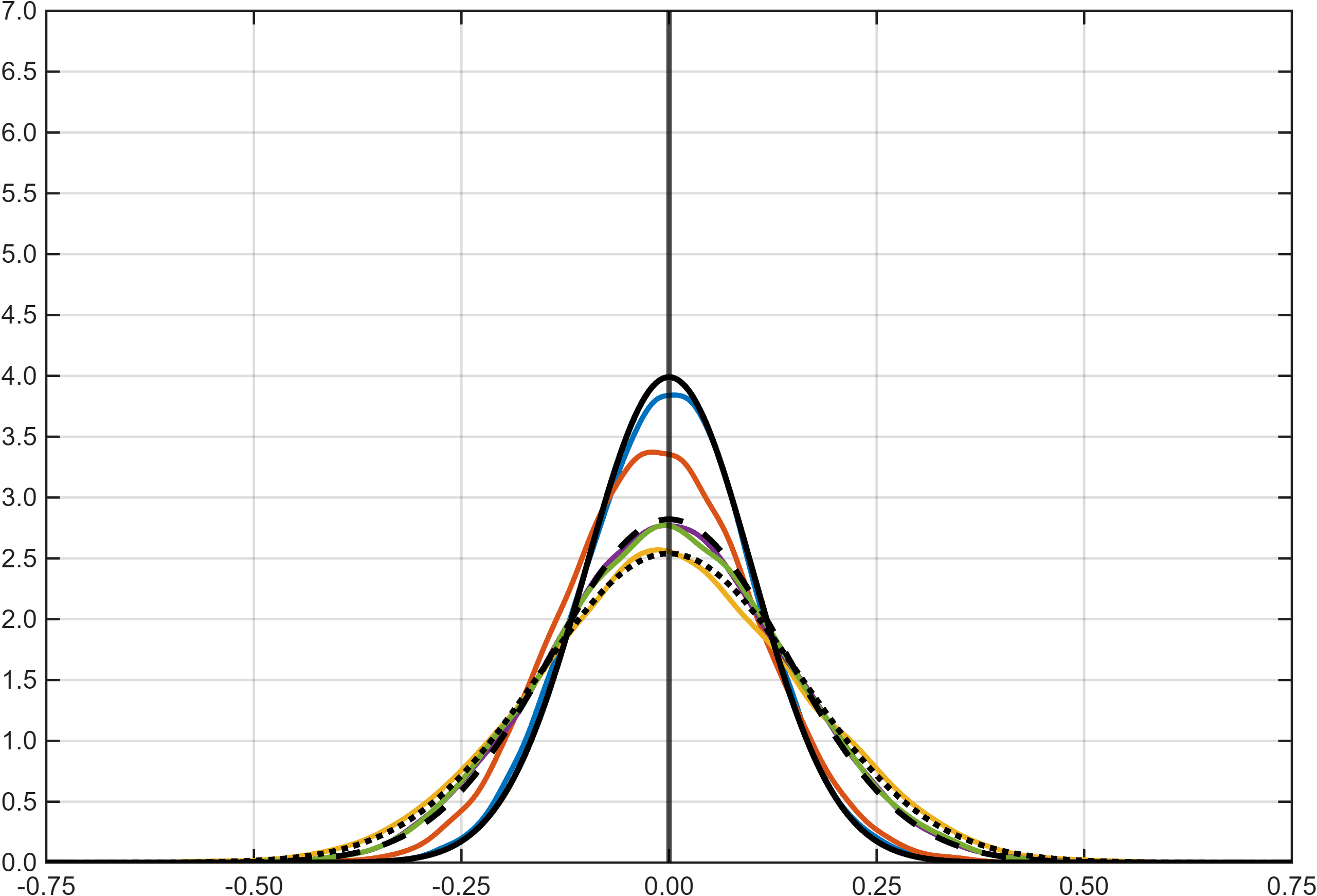}
\includegraphics[scale=0.24]{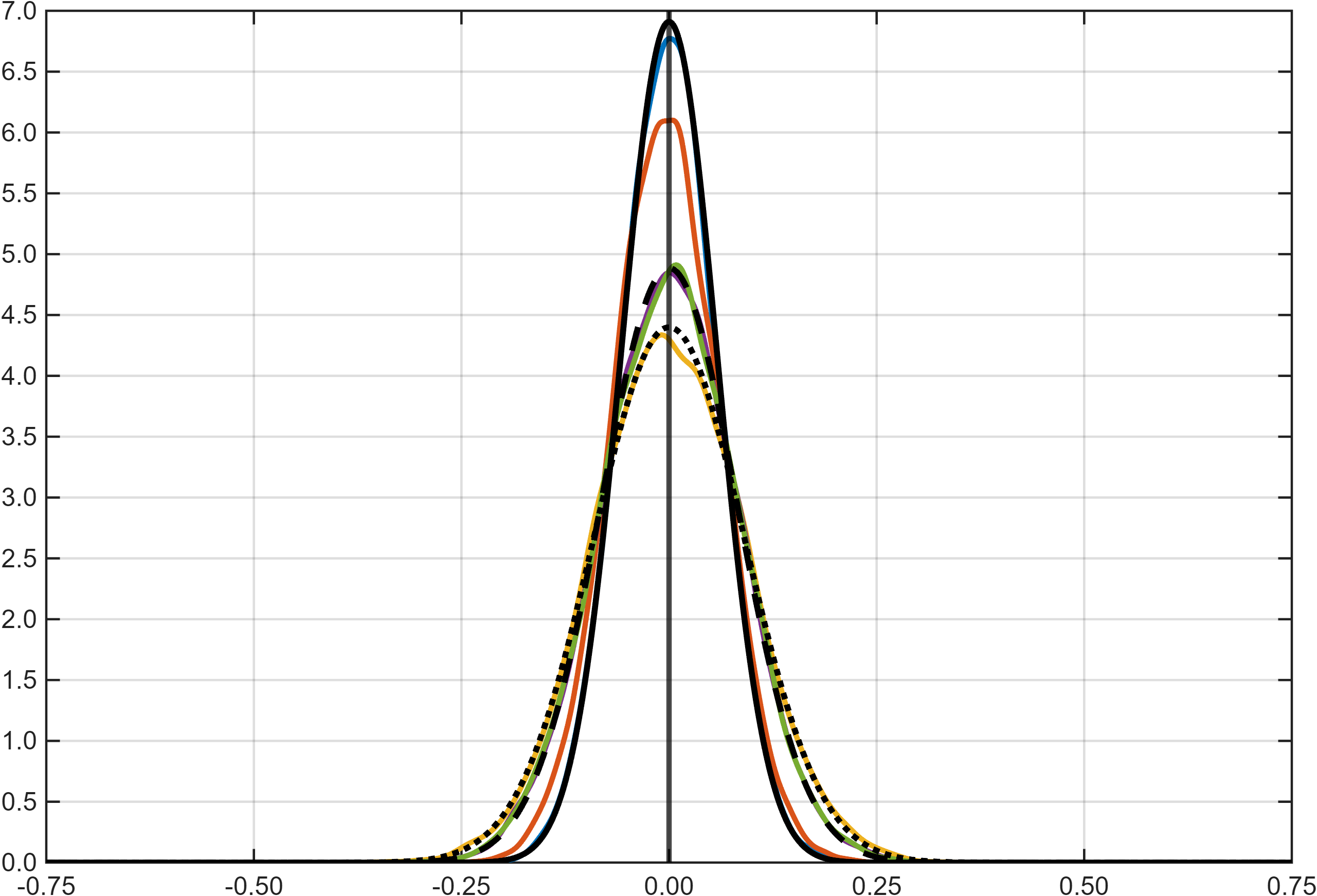}\\
\includegraphics[scale=0.24]{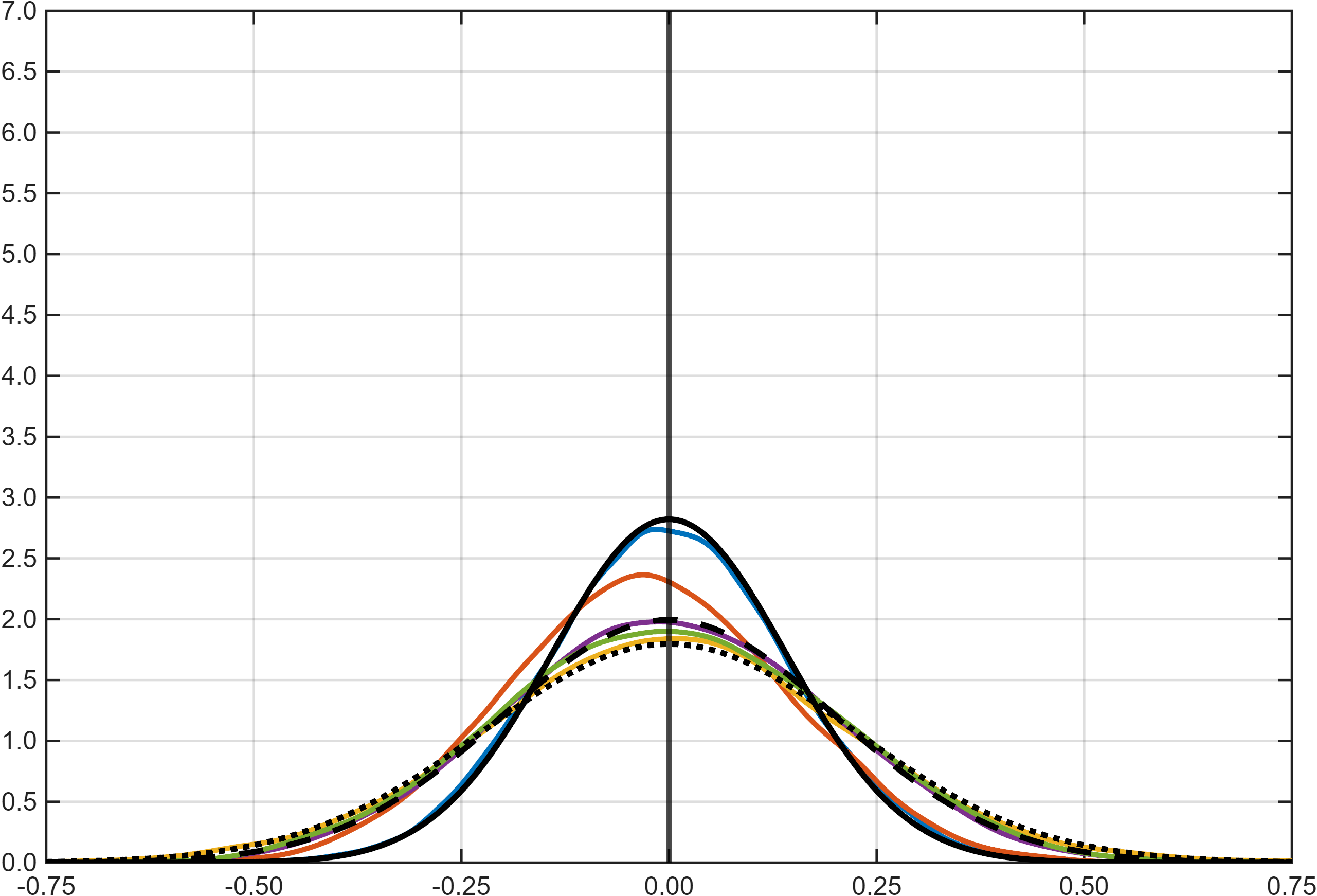}
\includegraphics[scale=0.24]{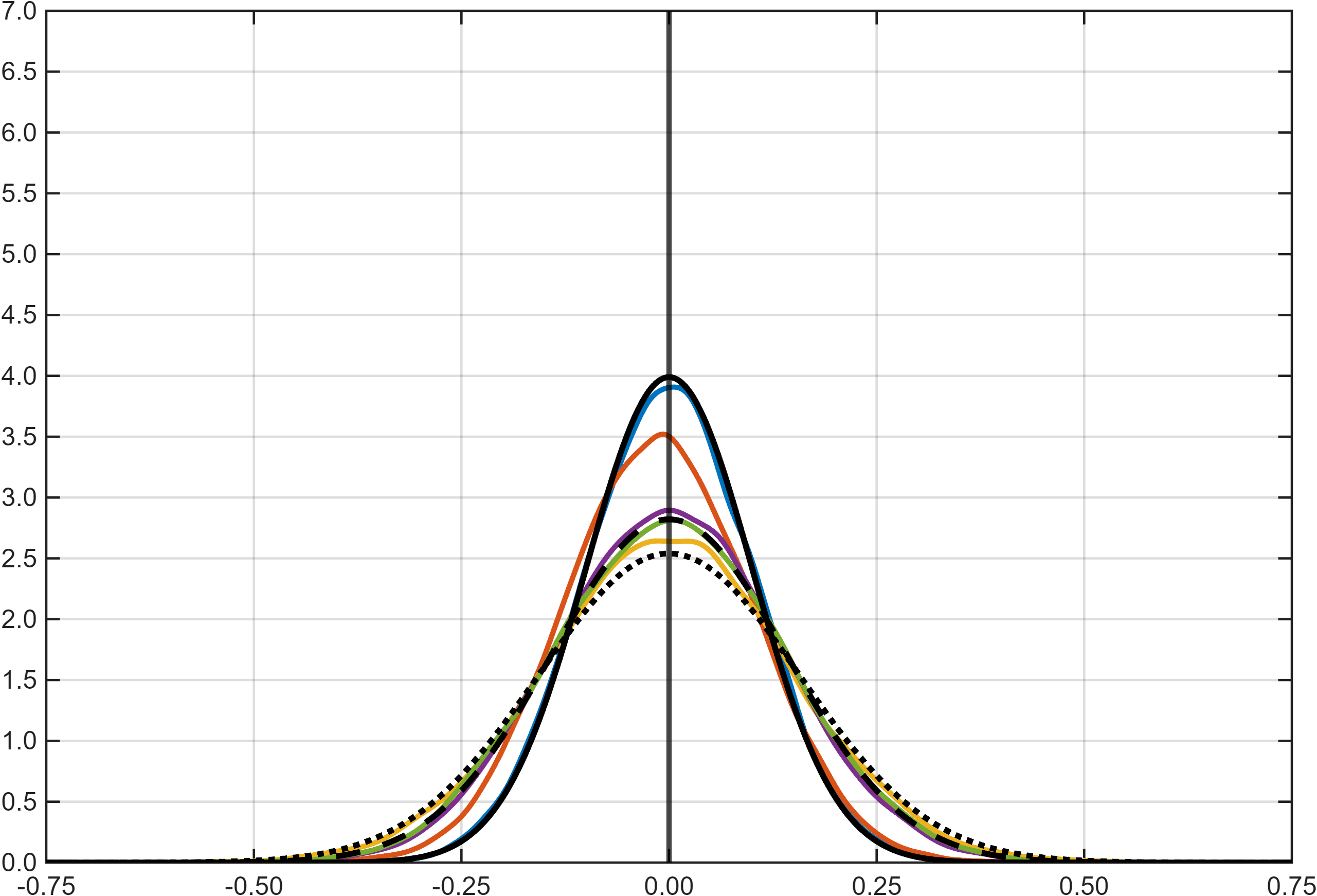}
\includegraphics[scale=0.24]{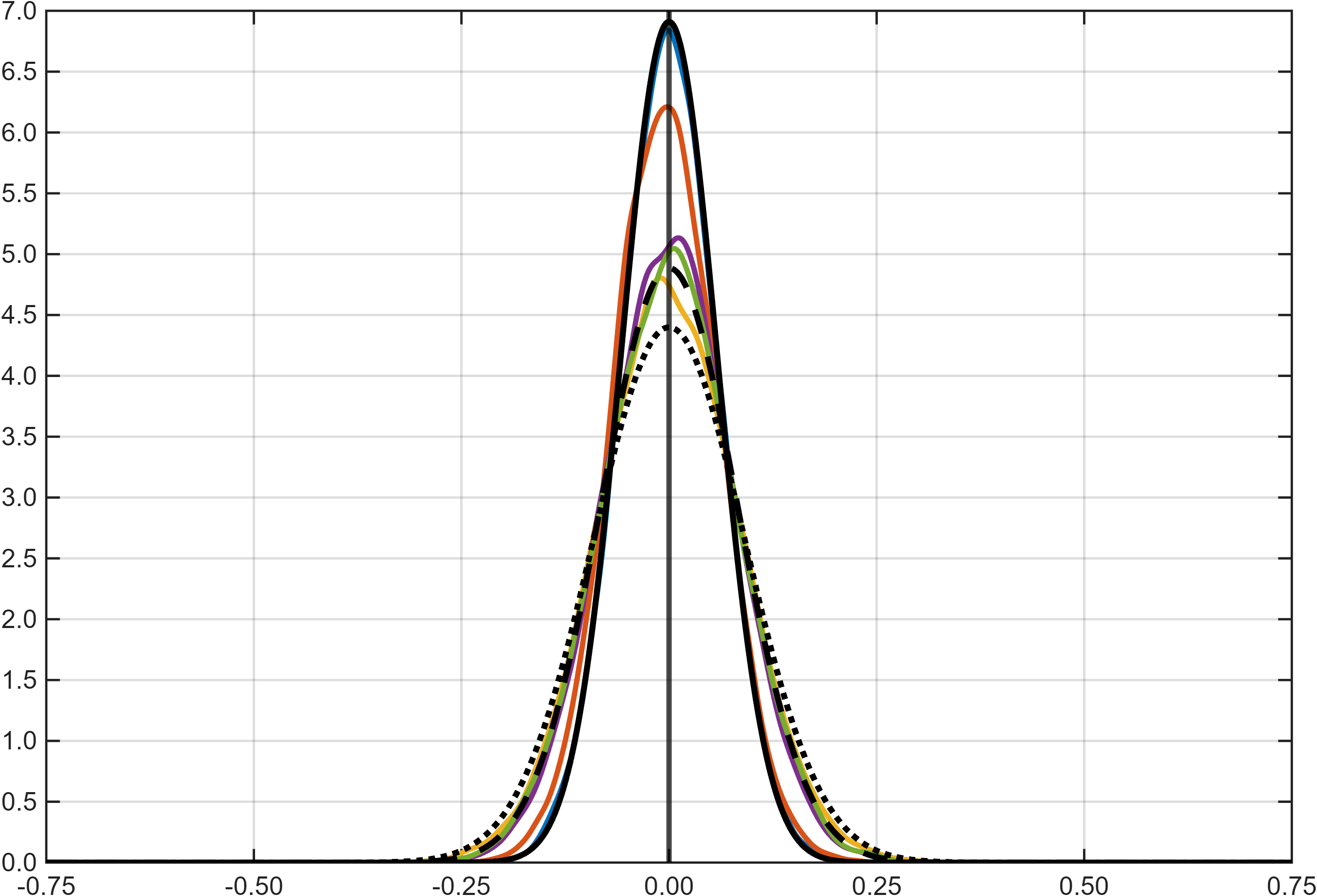}\\
\includegraphics[scale=0.24]{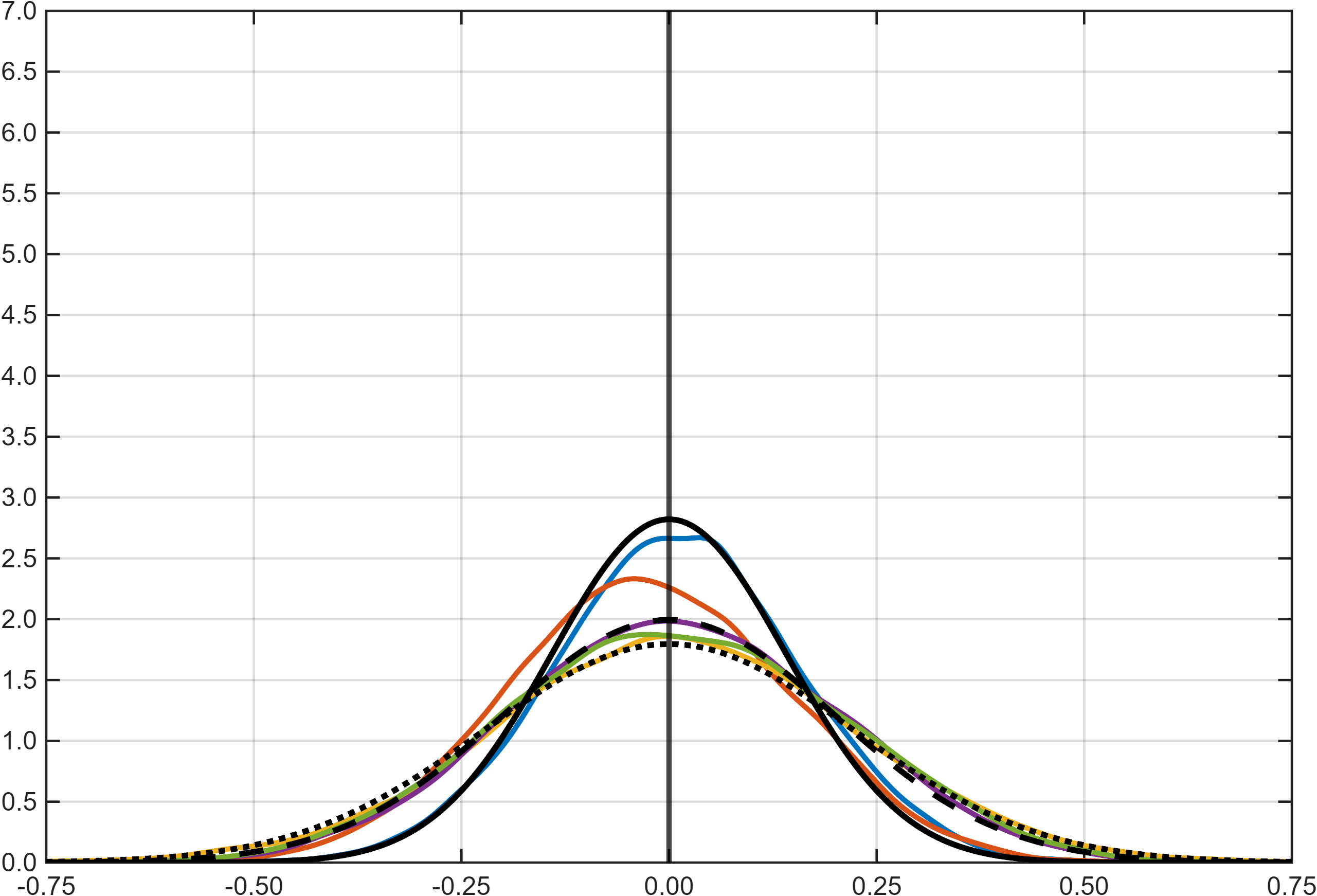}
\includegraphics[scale=0.24]{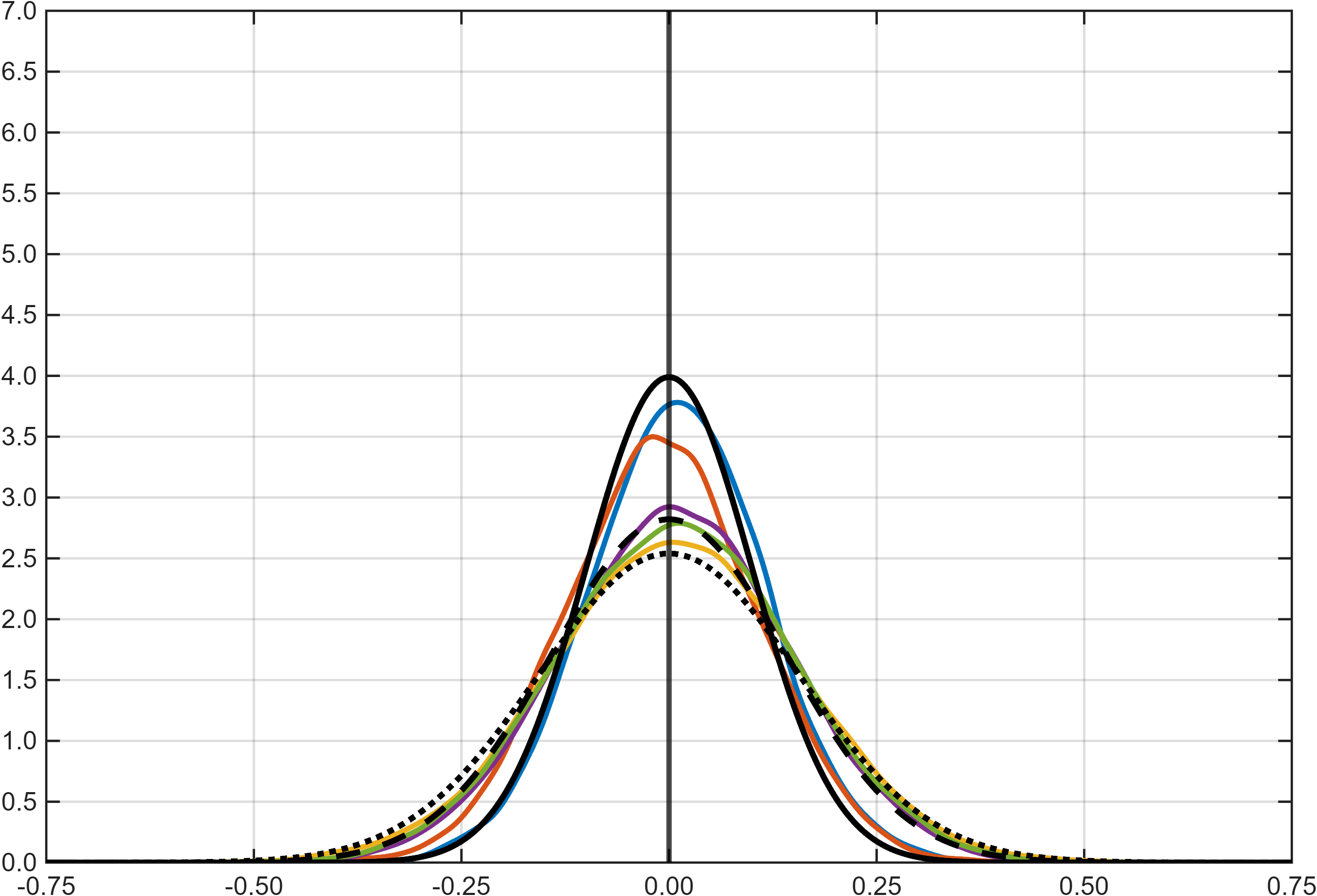}
\includegraphics[scale=0.24]{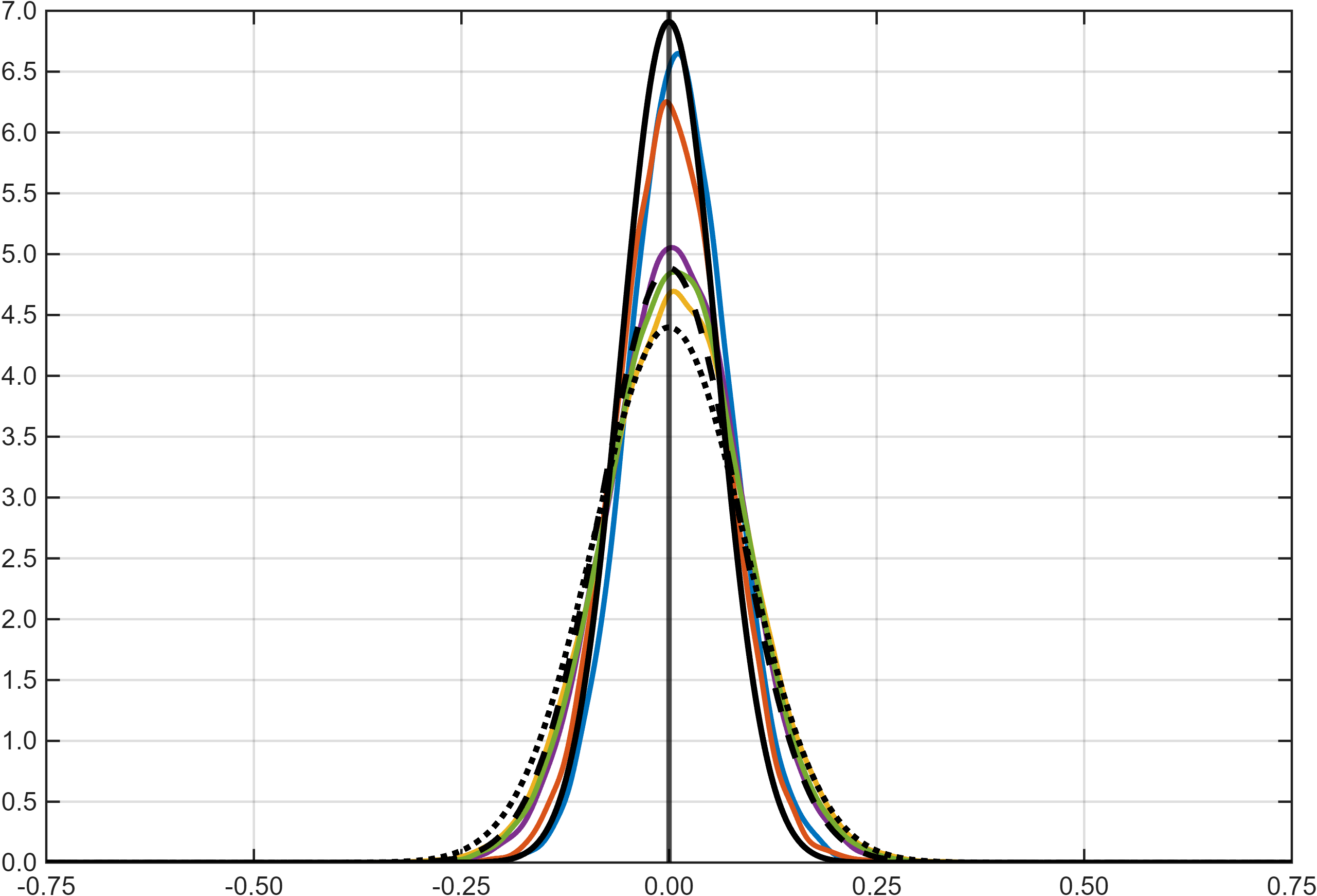}
\caption{Monte Carlo densities of estimators of $\rho(5)=0$ obtained when the simulated series of size $T=50$ (first column) and $T=100$ (second column) are generated with the AR(1) model with $\phi=0$ and contaminated with isolated additive outliers with $\delta=0$ (top row), $\delta=3$ and $\lambda=0$ (middle row), and $\delta=3$ and $\lambda=0.5$ (bottom row): $r(5)$ (blue), $\hat{\rho}_{MG}(5)$ (orange), $\hat{\rho}_H(5)$ (yellow), $\hat{\rho}_Q(5)$ (purple), and $\hat{\rho}_P(5)$ (green). The black densities represent the finite sample approximation of the asymptotic density of the sample autocorrelations, $N(0,\frac{1}{T})$, (continuous line), the H estimator, $N(0, \frac{\pi^2}{4T})$, and of the QML estimator, $N(0, \frac{2}{T})$, (dashed line).}
\label{fig:densities_2}
\end{figure}


\subsection{Finite sample performance of uncorrelatedness tests}

In this subsection, we carry out Monte Carlo experiments to analyse the finite sample performance of the B-P and L-B tests when implemented to test for $H_0: \rho(1)=\rho(2)=...=\rho(H)=0$ using either the sample autocorrelations or one of the robust estimators of the autocorrelation function considered in this paper. It is important to remark that the asymptotic distribution of the joint vector of estimated autocorrelations is only available when they are estimated using the sample autocorrelations. In the case, of the robust estimators, there is not a formal derivation of the joint asymptotic distribution of the estimator of the autocorrelations for different lags.

\begin{table}[th]
\caption{Monte Carlo results for the size of the B-P and L-B tests for $H=10$ lags when implemented to i) sample autocorrelations (SAC); ii) $\hat{\rho}_{H}(h)$ (H); iii) $\hat{\rho}_{ML}(h)$ (ML); and iv) $\hat{\rho}_{MR}(h)$ (MR),  when the series are Gaussian white noise contaminated by additive outliers of size $\delta=0, 3, 5$, with $\lambda=0, 0.5$. Nominal size: $5\%$.}
\label{tab:sim_size}
\begin{tabular}{lccccccccccccccccccccccccc}
\hline
& \multicolumn{5}{c}{B-P} & & \multicolumn{5}{c}{L-B}\\
\hline
& $\delta=0$ & \multicolumn{2}{c}{$\delta=3$} & \multicolumn{2}{c}{$\delta=5$}&& $\delta=0$ & \multicolumn{2}{c}{$\delta=3$} & \multicolumn{2}{c}{$\delta=5$} \\
\hline
 & & $\lambda=0$ & $\lambda=0.5$ & $\lambda=0$ & $\lambda=0.5$ & &  & $\lambda=0$ & $\lambda=0.5$ & $\lambda=0$ & $\lambda=0.5$\\
\hline
& \multicolumn{11}{c}{$T=50$}\\
 \hline
SAC &	0.060 & 0.059 & 0.120 & 0.044 & 0.211 && 0.117 & 0.115 & 0.190 & 0.088 & 0.287\\
H & 0.122 & 0.089 & 0.098 & 0.066 & 0.091 && 0.203 & 0.149 & 0.165 & 0.121 & 0.149\\
ML & 0.093 & 0.065 & 0.073 & 0.045 & 0.080 && 0.164 & 0.119 & 0.135 & 0.090 &	0.134\\
MR & 0.078 & 0.069 & 0.078 & 0.055 & 0.105 &&	0.149 & 0.131 & 0.146 & 0.113 & 0.176\\
\hline
& \multicolumn{11}{c}{$T=100$}\\
 \hline
SAC & 0.069 & 0.062 & 0.214 & 0.058 & 0.494 && 0.094 & 0.086 & 0.255 & 0.079 & 0.536\\
H & 0.093 &	0.057 & 0.080 & 0.031 & 0.064 && 0.125 & 0.077 &	0.106 & 0.045 & 0.085\\
ML & 0.069 & 0.045 & 0.073 & 0.030 & 0.085 && 0.099 & 0.060 &	0.099 & 0.043 & 0.108\\
MR & 0.063 & 0.049 & 0.081 & 0.041 & 0.112 && 0.088 &	0.074  & 0.116 & 0.066 &	0.151\\
\hline
& \multicolumn{11}{c}{$T=200$}\\
 \hline
SAC &  0.065 & 0.057 & 0.431 &	0.058 & 0.843 && 0.075 &	 0.068 & 0.453 & 0.069 &	0.853\\
H & 0.075 &	0.040 & 0.084 & 0.017 & 0.097 && 0.090 & 0.048 &	0.096 & 0.020 & 0.108\\
ML & 0.059 & 0.037 & 0.088 & 0.022 & 0.130 && 0.071 & 0.044 &	0.101 & 0.029 & 0.146\\
MR & 0.057 & 0.048 & 0.103 & 0.040 & 0.178 && 0.069 & 0.058 & 0.119 & 0.049 &	0.196\\
\hline
\end{tabular}
\end{table}

Table \ref{tab:sim_size} reports the Monte Carlo empirical size when the series are simulated by Gaussian white noise contaminated by additive outliers with the same design explained above. We can observe that, regardless of whether the series are contaminated or not, and of the sample size, the B-P and L-B tests are mostly oversized, with the size of B-P being always closer to the nominal $5\%$ than the size of the L-B test; see, for example, Kwan and Sim (1996), who show that the B-P test suffers of location bias, while the L-B test is still questionable. Therefore,  we focus the analysis on the B-P test. When there is no contamination, the B-P test based on $\hat{\rho}_{P}(h)$ has power closer to the nominal among tests based on other alternative estimators of the autocorrelations considered. Only when $T=50$, the over-rejection is larger than that based on the sample autocorrelations. Similarly, when there is contamination with $\lambda=0$, the conclusions are similar with the B-P test having sizes closer to that observed when there are not outliers. However, the B-P test based on sample autocorrelations is useless when the contamination is such that $\lambda=0.5$, with empirical sizes that can be as large as $84\%$ when $T=200$ and $\delta=5$. The distortion of the size is moderate when the B-P test is based on $\hat{\rho}_{P}(h)$. 

\begin{table}[th]
\caption{Monte Carlo results for the power of the B-P and L-B tests for $H=10$ lags when implemented to i) sample autocorrelations (SAC); ii) $\hat{\rho}_{H}(h)$ (H); iii) $\hat{\rho}_{ML}(h)$ (ML); and iv) $\hat{\rho}_{MR}(h)$ (MR),  when the series, generated by an AR(1) model with parameter $\phi=0.3, 0.7$ and the innovations are Gaussian white noise, are contaminated by additive outliers of size $\delta=0, 3, 5$, with $\lambda=0, 0.5$. Nominal size: $5\%$.}
\label{tab:sim_power}
\begin{tabular}{lccccccccccccccccccccccccc}
\hline
& \multicolumn{5}{c}{B-P} & & \multicolumn{5}{c}{L-B}\\
\hline
& $\delta=0$ & \multicolumn{2}{c}{$\delta=3$} & \multicolumn{2}{c}{$\delta=5$}&& $\delta=0$ & \multicolumn{2}{c}{$\delta=3$} & \multicolumn{2}{c}{$\delta=5$} \\
\hline
 & & $\lambda=0$ & $\lambda=0.5$ & $\lambda=0$ & $\lambda=0.5$ & &  & $\lambda=0$ & $\lambda=0.5$ & $\lambda=0$ & $\lambda=0.5$ \\
\hline
& \multicolumn{11}{c}{$T=50$}\\
 \hline
& \multicolumn{11}{c}{$\phi=0.3$}\\
\hline			
SAC & 0.251 & 0.143 & 0.362 & 0.087 & 0.419 && 0.346 &	0.217 & 0.452 & 0.142 & 0.515\\
H & 0.199 &	0.145 & 0.203 & 0.114 & 0.177 && 0.289 & 0.214 & 0.288 & 0.173 & 0.252\\
ML & 0.185 & 0.127 & 0.205 & 0.100 & 0.206 && 0.274 & 0.194 & 0.291 & 0.163 & 0.291\\
MR & 0.178 & 0.135 & 0.221 & 0.128 & 0.256 && 0.279 & 0.223 & 0.321 & 0.208 & 0.352\\
\hline
& \multicolumn{11}{c}{$\phi=0.7$}\\
\hline
SAC &	 0.950 & 0.654 & 0.875 &	0.372 & 0.807 && 0.967 & 0.715 & 0.908 & 0.436 &	0.851\\
H & 0.696 &	0.564 & 0.606 & 0.544 & 0.561 && 0.761 & 0.632 & 0.682 & 0.611 & 0.634\\
ML & 0.766 & 0.629 & 0.718 & 0.616 & 0.721 && 0.822 & 0.700 & 0.783 & 0.684 & 0.785\\
MR & 0.787 & 0.671 & 0.749 & 0.690 & 0.775 && 0.843 &	0.752 & 0.820 & 0.768 & 0.843\\
\hline
& \multicolumn{11}{c}{$T=100$}\\
 \hline
& \multicolumn{11}{c}{$\phi=0.3$}\\
\hline
SAC & 0.484 & 0.253 & 0.702 & 0.128 & 0.808 && 0.528 &	0.299 & 0.738 & 0.160 & 0.833\\
H & 0.093 &	0.057 & 0.080 & 0.031 & 0.064 && 0.306 & 0.196 & 0.309 & 0.140 & 0.311\\
ML & 0.069 & 0.045 & 0.073  & 0.030 & 0.085 && 0.328 &	0.213 & 0.387 & 0.178 & 0.439\\
MR & 0.063 & 0.052 & 0.087 & 0.044 & 0.119 && 0.338 & 0.253 & 0.432 & 0.239 & 0.509\\
\hline
& \multicolumn{11}{c}{$\phi=0.7$}\\
\hline
SAC &	1.000 & 0.923 & 0.997 & 0.627 & 0.988 && 1.000 & 0.932 & 0.998 & 0.659 & 0.990\\
H & 0.948 &	0.872 & 0.905 & 0.846 &	0.886 && 0.956 & 0.890 & 0.920 & 0.862 & 0.904\\
ML & 0.981 & 0.936 & 0.972 & 0.929 & 0.978 && 0.985 & 0.946 & 0.977 & 0.942 & 0.983\\
MR & 0.982 & 0.952 & 0.978 & 0.958 & 0.986 && 0.987 & 0.962 & 0.984 & 0.967 & 0.990\\
\hline
& \multicolumn{11}{c}{$T=200$}\\
 \hline
& \multicolumn{11}{c}{$\phi=0.3$}\\
\hline
SAC & 0.831 & 0.471 & 0.970 & 0.217 & 0.994 && 0.844 & 0.494 & 0.973 & 0.237 &	0.994\\
H & 0.423 &	0.261 & 0.530 & 0.182 &	0.579 && 0.445 & 0.279 & 0.554 & 0.200 & 0.601\\
ML & 0.509 & 0.337 & 0.671 & 0.299 & 0.758 && 0.537 & 0.360 & 0.691 & 0.320 & 0.774\\
MR & 0.516	& 0.385 & 0.704 & 0.393 & 0.801 && 0.539 & 0.412 & 0.724 & 0.419 &	0.817\\
\hline
& \multicolumn{11}{c}{$\phi=0.7$}\\
\hline
SAC & 1.000 & 0.999 & 1.000 &	0.906 & 1.000 && 1.000 & 0.999 & 1.000 & 0.912 &	1.000\\
H & 1.000 &	0.995 & 0.999 & 0.994 & 0.999 && 1.000 & 0.995 &	0.999 & 0.994 & 0.999\\
ML & 1.000 & 0.999 & 1.000 & 1.000 & 1.000 && 1.000 & 0.999 &	1.000 & 1.000 & 1.000\\
MR & 1.000 & 1.000 & 1.000 & 1.000 & 1.000 && 1.000 & 1.000 & 1.000 & 1.000 & 1.000\\
 \hline
\end{tabular}
\end{table}

The size distortions are not due to the way the correlations are computed but to the tests themselves. Therefore, alternative modifications of the B-P and L-B tests can be considered; see, for example, Kwan and Sim (1996). We left this analysis for further research.

\section{Empirical illustrations}
\label{sec:empiric}

In this section, we illustrate the usefulness in empirical applications of the robust estimators of the autocorrelation function proposed in this paper. We consider estimating the autocorrelations of a series of daily financial returns of the IBEX35, a series of quarterly US economic growth, and monthly US inflation. 

\subsection{Financial returns}

Consider closing prices of the IBEX35 index of Madrid Stock Exchange observed daily from 1/1/2019 to 12/8/2025. Figure \ref{fig:returns} plots the returns, obtained as usual as $y_t=100 \times \left(\log{P_t} - \log{P_{t-1}} \right)$, where $P_t$ denotes the closing price of day $t$, together with their corresponding estimated sample autocorrelations and their 95\% non-rejection bands. Note that IBEX35 returns have a significant second order autocorrelation, which implies that IBEX35 daily price movements are predictable and consequently, that the Efficient Market Hypothesis (EMH) is violated; see, for example, the description of tests for EMH by Reschenhofer and Hauser (1997).\endnote{The EMH, developed by Bachelier (1900) and Fama (1965), states that in an efficient market, securities will be fairly priced to reflect all available information.}

\begin{figure}[ht!]
\begin{center}
\includegraphics[clip=true, trim=0 0 0 0cm, scale=0.45]{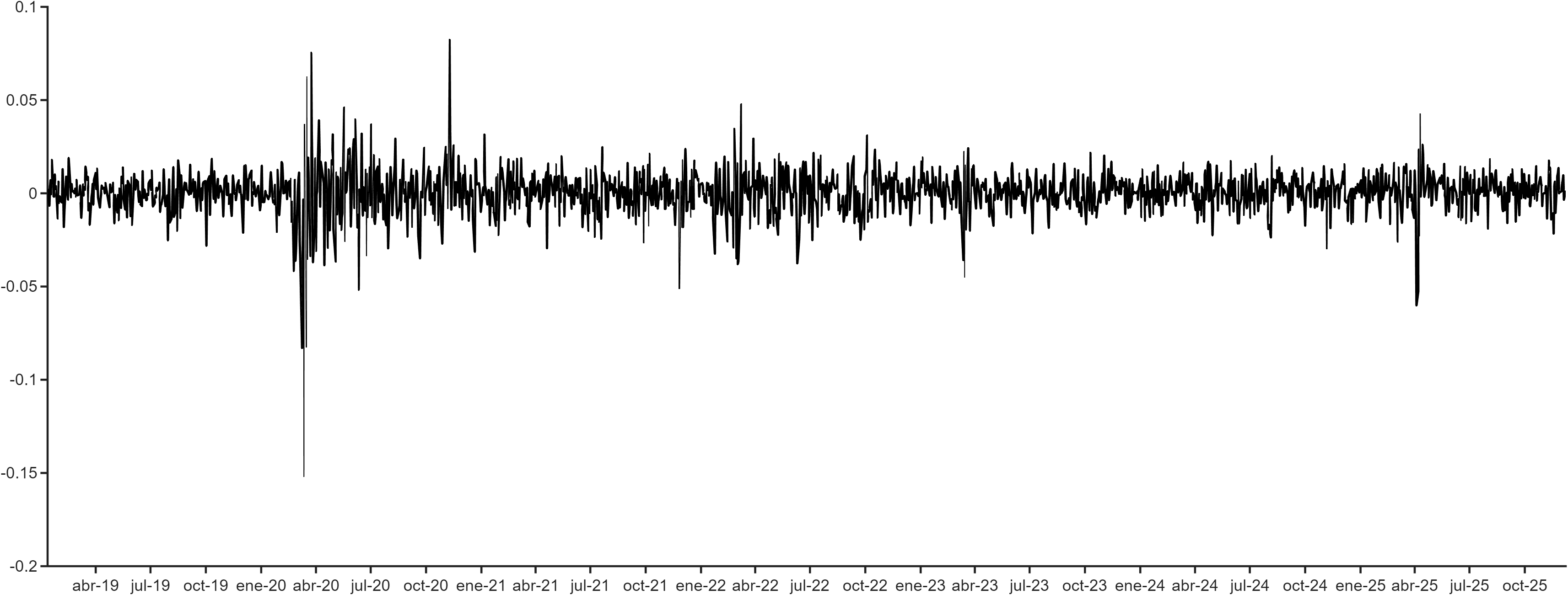}\\
\includegraphics[clip=true, trim=0 0 0 0cm, scale=0.45]{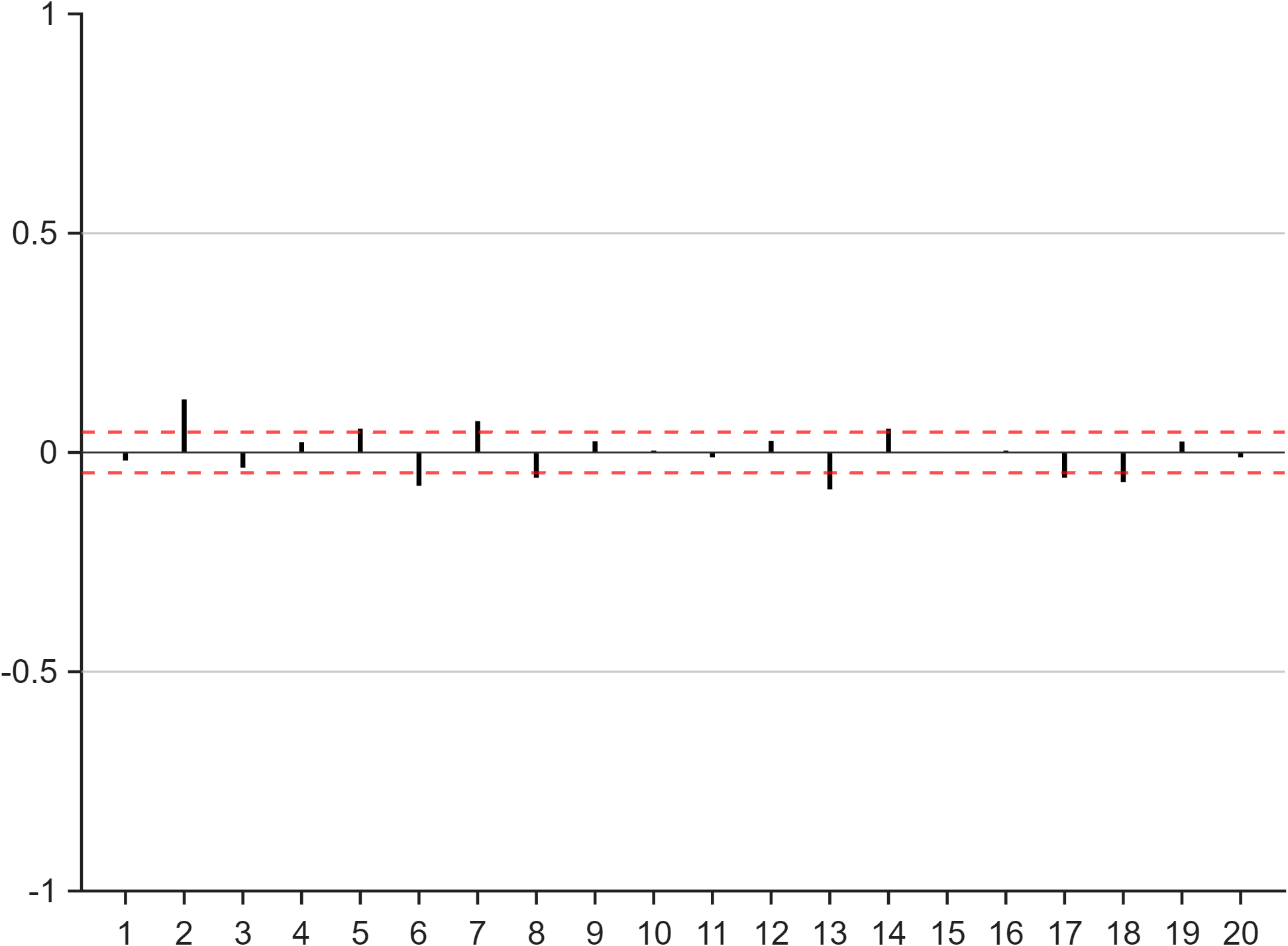}
\includegraphics[clip=true, trim=0 0 0 0cm, scale=0.45]{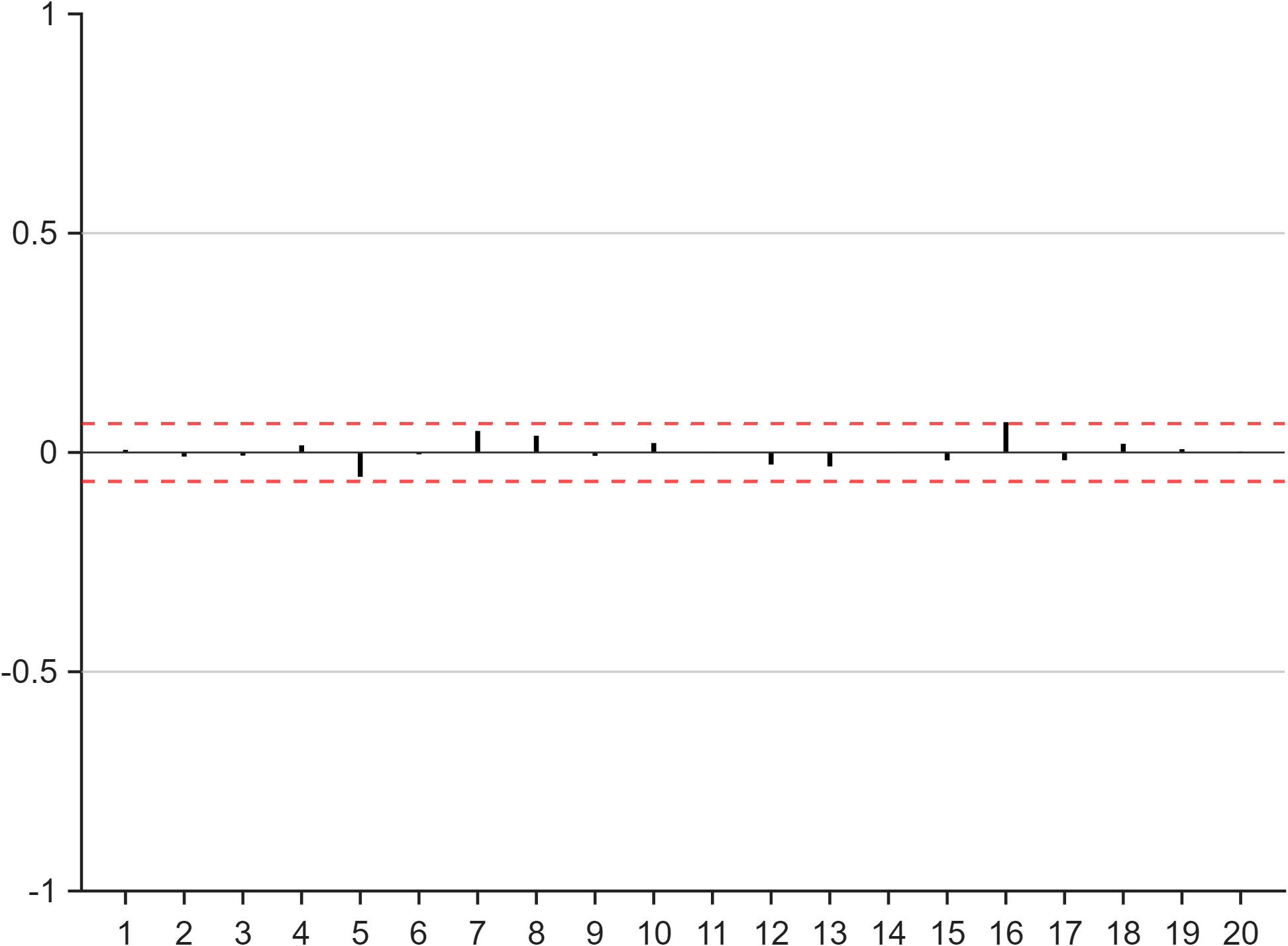}
\end{center}
\caption{Daily IBEX35 returns (upper panel), and estimated sample autocorrelations (lower left panel) and robust plug-in autocorrelations (lower right panel) together with 95\% non-rejection bands (red lines).}
\label{fig:returns}
\end{figure}

Figure \ref{fig:returns} also shows that daily IBEX35 returns observed over the last seven years have several extreme observations, which could be thought as potential outliers.\endnote{As usual when analysing daily financial returns, IBEX35 returns seem to be conditionally heteroscedastic. We conjecture that the robust estimator of the autocorrelations proposed in this paper is also robust to the presence of conditional heteroscedasticity. However, this analysis is left for further research.} Some of them (9, 12 and 16 of March, 2020) can be associated to the effect of the COVID19 pandemic on the Spanish financial market, with a clear increase of the market volatility associated to the implementation of extraordinary policies aimed to stop the virus propagation. We can also observe an unusually large positive return on the 9th of November 2020, which is when preliminary results of the Phase 3 clinical trials of the Pfizer-BioNTech were published, reporting an effectiveness over 90\%, which trigger a positive increase in financial markets. Finally, on the 7th and 10th of April, 2025, we can also observe extreme movements in IBEX35 prices, which can be linked to the announcement of new customs tariffs by the Trump Administration and the consequent announcement of their softening, respectively. To take into account the potential presence of outliers on the estimated autocorrelations of IBEX35 returns, we also estimate them by using the robust estimators described in this paper. Figure \ref{fig:returns} plots the autocorrelations estimated using the plug-in estimator proposed in this paper.\endnote{Other alternative robust estimators provide similar estimated autocorrelations, which are available upon request.} The autocorrelations plotted in Figure \ref{fig:returns} are not significant for any lag, with the conclusion that there are not linear dependencies in the daily movements of financial IBEX35 returns, and consequently supporting the EMH. Similar results are obtained when testing for joint significance of the first 10 autocorrelations. In this case, the B-P statistic takes value 61.31 with a \textit{p}-value of 0.00, which implies a rejection of the null hypothesis and, consequently linear dependence of returns. However, the robust version of the B-P statistic based on the proposed plug-in estimator takes a value of 7.03 with a \textit{p}-value of 0.72, which is not significant and, consequently, there is not evidence against the null hypothesis of zero autocorrelations of IBEX35 returns.

\subsection{US economic growth}

\begin{figure}[ht!]
\begin{center}
\includegraphics[clip=true, trim=0 0 0 0cm, scale=0.45]{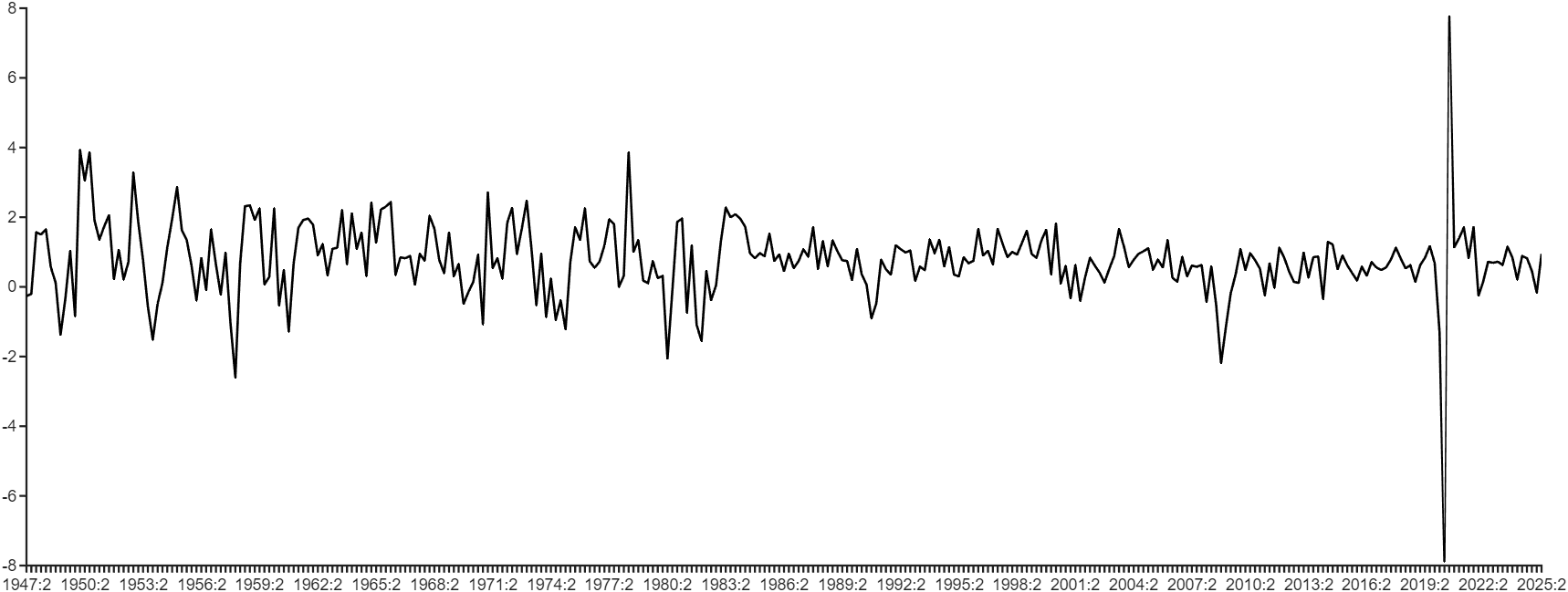}\\
\includegraphics[clip=true, trim=0 0 0 0cm, scale=0.45]{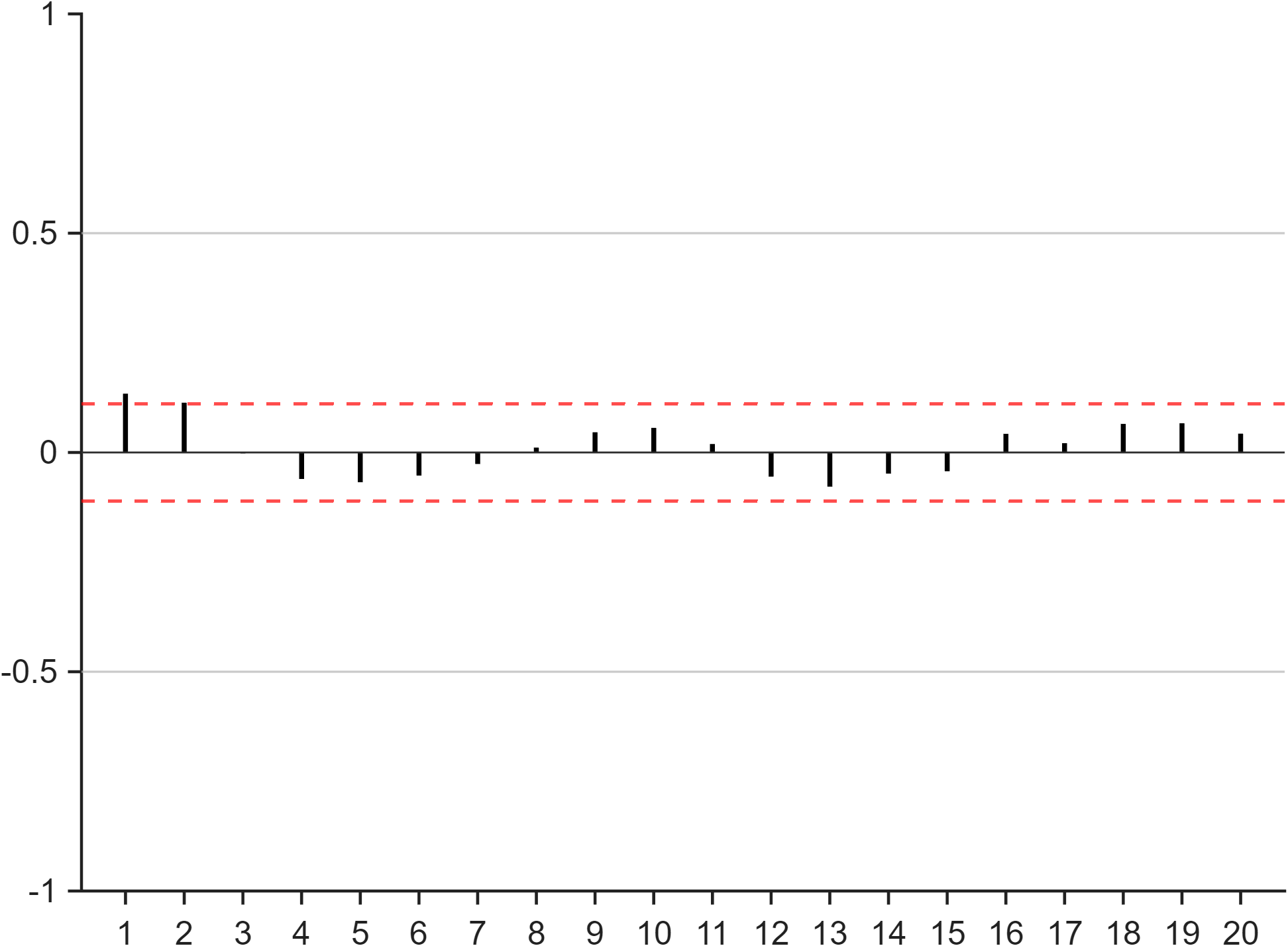}
\includegraphics[clip=true, trim=0 0 0 0cm, scale=0.45]{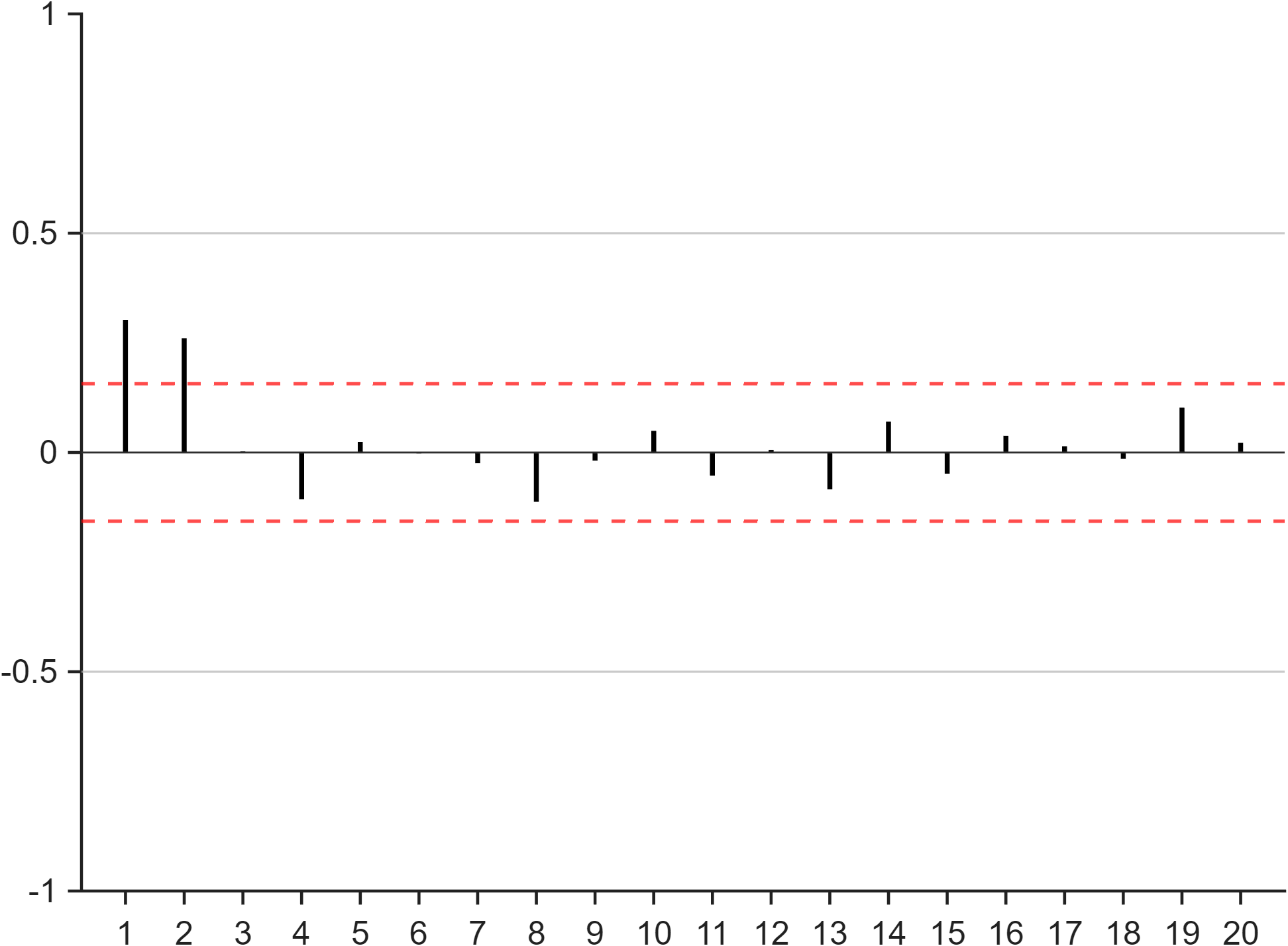}
\end{center}
\caption{Quarterly US growth (upper panel), and estimated sample autocorrelations (lower left panel) and robust plug-in autocorrelations (lower right panel) together with 95\% non-rejection bands (red lines).}
\label{fig:growth}
\end{figure}

A proper economic growth forecast is crucial because it guides major decisions for governments (fiscal/monetary policy, budgeting), businesses (investment, hiring, production), and individuals (spending, saving), fostering stability and confidence. However, quarterly GDP data are still published with a delay that ranges from 30 days after the end of the reference quarter in the USA, to 70 days in the main Euro area countries. Consequently, there is a clear interest in obtaining accurate short-run forecasts of GDP growth. One problem when forecasting GDP growth in the short-run though is that in most countries its serial correlation is very weak; this means that traditional univariate linear time series models tend to not be very useful tools for this purpose; see, for example, Golinelli and Parigi (2007). Consider GDP annual growth in the US observed quarterly from 1947Q2 to 2025Q2, obtained from the Federal Reserve Economic Data (FRED) of the Federal Reserve Bank of St. Louis and seasonally adjusted. Figure \ref{fig:growth} plots the US GDP growth calculated as $y_t=100 \times \frac{GDP_t-GDP_{t-1}}{GDP_{t-1}}$ together with the estimated sample autocorrelations and their 95\% non-rejection bands. According to the estimated sample autocorrelations, only the first order one is just out of the non-rejection bands with a relatively small magnitude around 0.1. This is in concordance with the results in the literature about weak linear temporal dependence in economic growth. Alternatively, Figure \ref{fig:growth}, which plots the autocorrelations estimated using the plug-in estimator, shows a very different picture.\endnote{Note that, as above, the estimated autocorrelations estimated using alternative robust estimators are very similar.} Not only the first two autocorrelations are significant, but also their magnitudes are clearly larger than before with values larger than 0.2, implying a stronger temporal dependence of growth. Similarly, when computing the B-P statistic for testing the joint significance of the first 10 autocorrelations based on the sample autocorrelations, the value is 15.06 with a \textit{p}-values of 0.13. Therefore we do not reject the null, implying no linear temporal dependence. However, when the B-P statistic is calculated based on the estimated robust plug-in autocorrelations, the value is 29.29 with a \textit{p}-value of 0.00, and consequently, the null hypothesis is clearly rejected.

\subsection{US inflation}

There is a long-lasting debate on the persistence characteristics of inflation, with several authors claiming its stationarity (Pivetta and Reis, 2007, and Fuhrer, 2010), others modelling it as a non-stationary process (Barsky. 1987, and Ball, Cecchetti and Gordon, 1990), while a last group of works consider inflation as a long-memory process (Hassler and Wolters, 1995, Gadea and Mayoral, 2006). In some of these works, the shape of the sample autocorrelation function is used to argue about the persistence characteristics of inflation; see, for example, Proietti and Luati (2026). We examine m-o-m inflation in the US observed  monthly  from January 1960 to June 2026, which has been downloaded from the FRED database at the Federal Reserve Bank of St Louis (https://fred.stlouisfed.org/series/CPIAUCSL). The value corresponding to October 2025 was missing and has been linearly interpolated. The series may be contaminated by several outliers as those corresponding to August 1973  during the period of control prices in the US, September 2005 with the huge increase in energy prices after the Katrina and Rita hurricanes, October to December 2008 during the global financial crisis, and April 2020 with the sudden fall in prices after the Covid-19 lock-down.   

The seasonally adjusted series of inflation is plotted in Figure \ref{fig:inflation} together with the autocorrelations estimated using their sample counterparts and using the robust plug-in estimator. We can observe that the estimated autocorrelations obtained using the plug-in estimator are larger in magnitude and imply a larger persistence than when they are estimated using the sample autocorrelations. In this case, the B-P statics for the joint significance of the first 10 autocorrelations are 1,149.21 (\textit{p}-value 0.00) and 1,486.34 (\textit{p}-value 0.00) and, obviously, in both cases the null is clearly rejected. Whether the stronger and more persistent autocorrelations may have implications for the decision about whether inflation is long-memory or non-stationary is left for further research.\endnote{Note that several authors have pointed out the relationship between long-memory and structural change, showing how changes in the conditional mean or variance can cause the autocorrelation function to behave the same way they would in a stationary long memory process; see, for example, Granger and Hyung (1999), Diebold and Inoue (2001) and Mikosch and Starica (2004). However, as far as we know, the implications of contamination by outliers on the determination of the persistence properties of a time series has not been analysed.} 

\begin{figure}[ht!]
\begin{center}
\includegraphics[clip=true, trim=0 0 0 0cm, scale=0.45]{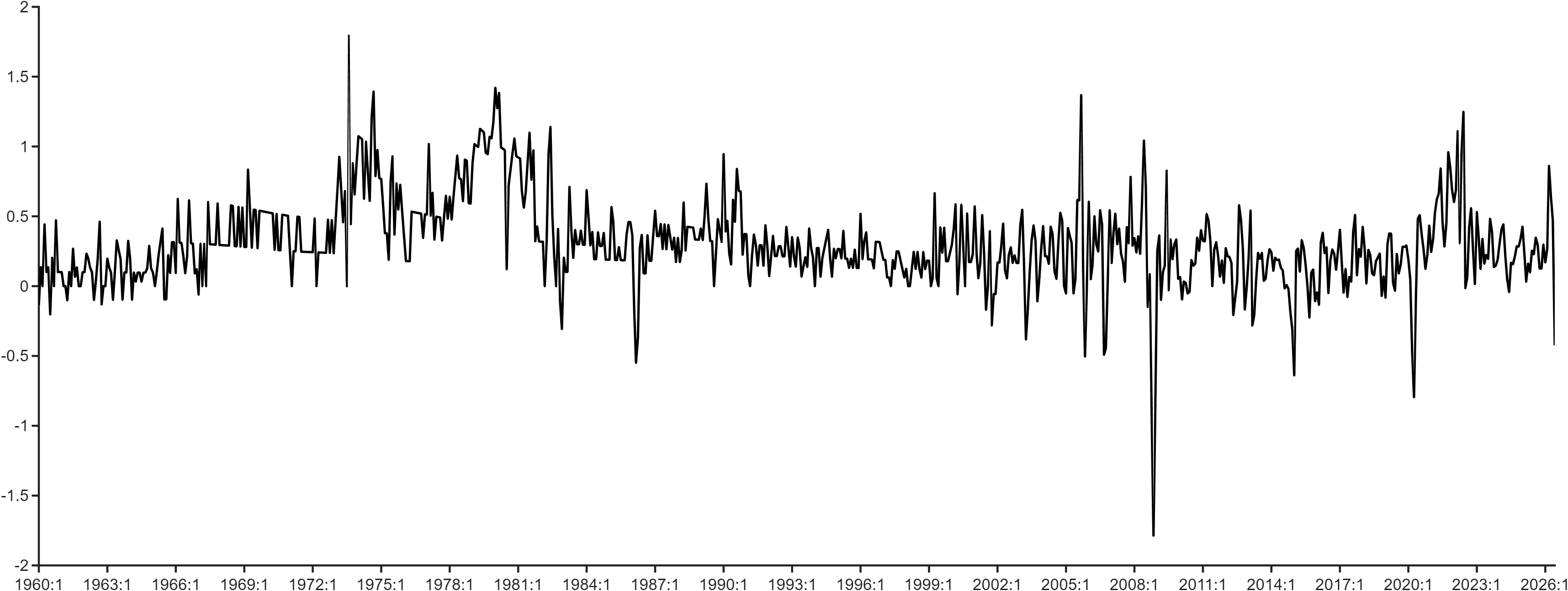}\\
\includegraphics[clip=true, trim=0 0 0 0cm, scale=0.45]{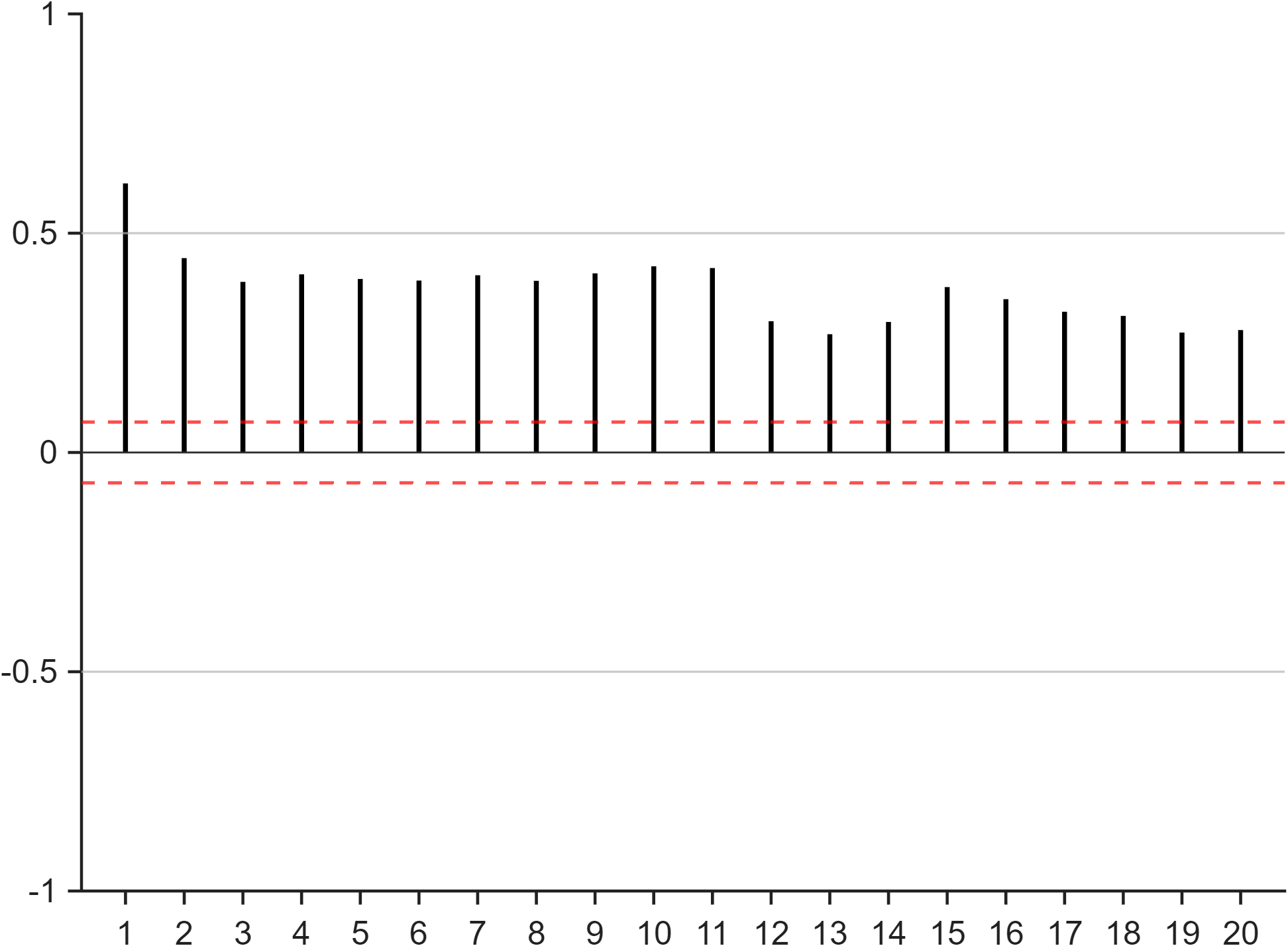}
\includegraphics[clip=true, trim=0 0 0 0cm, scale=0.45]{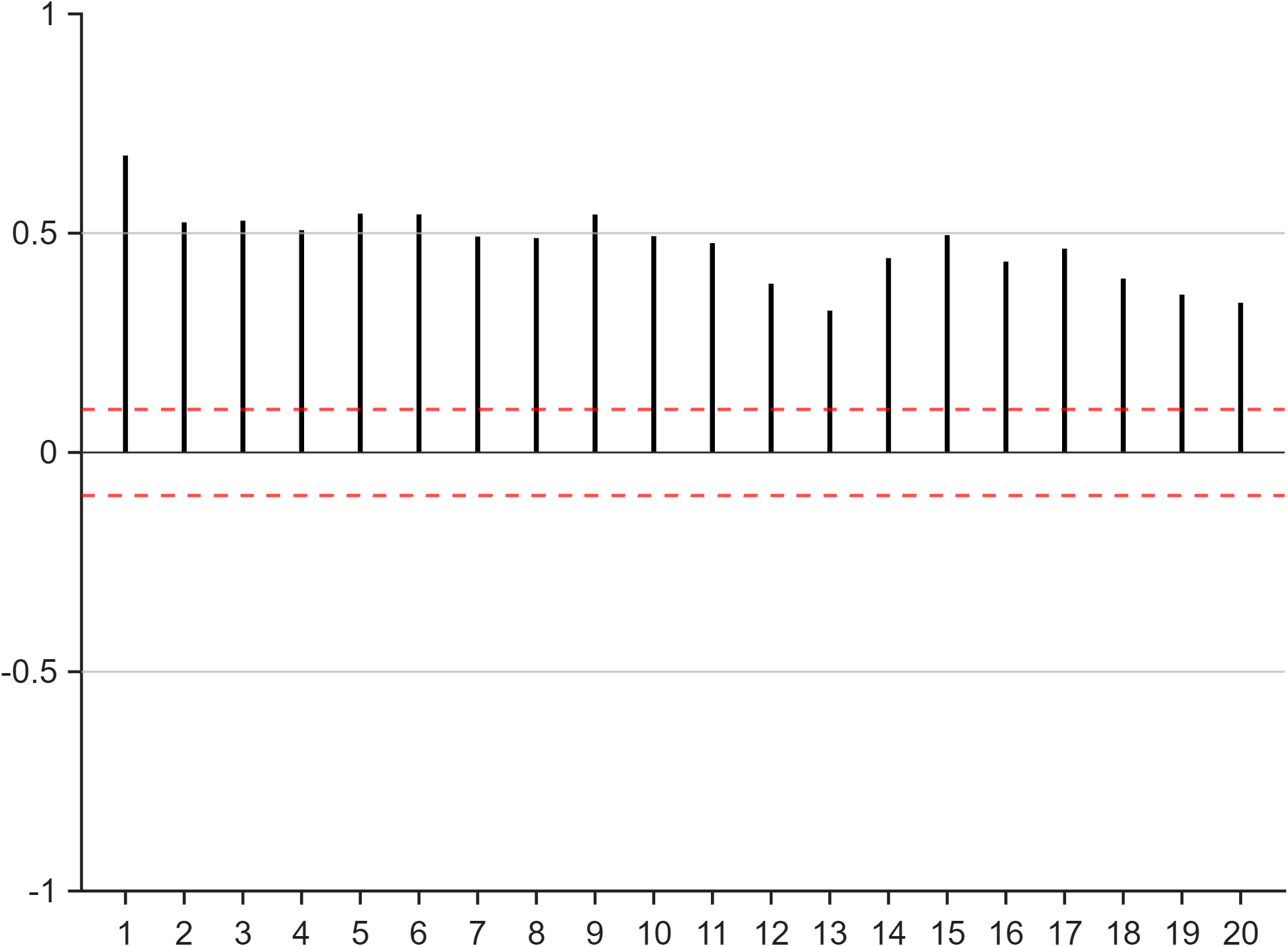}
\end{center}
\caption{Quarterly US inflation (upper panel), and estimated sample autocorrelations (lower left panel) and robust plug-in autocorrelations (lower right panel) together with 95\% non-rejection bands (red lines).}
\label{fig:inflation}
\end{figure}

\section{Conclusions}
\label{sec:conclusions}

Estimating the autocorrelations of a time series using the sample autocorrelations is highly affected by outliers, making it relevant to use robust methodologies in empirical applications in which time series may be contaminated. We propose three estimators of the autocorrelation function that are based on ratios of consecutive observations separated by $h$ periods of time, and are robust against outliers. The first estimator is based on the median and consequently, is not efficient. The second one is a more efficient QML estimator that requires numerical optimization. Finally, the third estimator is a plug-in estimator, which does not require computational optimization and have only slightly larger variance compared to the QML estimator. Consequently, the plug-in estimator is recommended to estimate the autocorrelation function. The asymptotic distributions of the three proposed estimators are derived when the true autocorrelations are zero. In doing so, we allow for the construction of point-wise significant bounds and individual tests of significance based on the proposed estimators. 

Based on simulations, we compare the performance of the  plug-in estimator with an alternative popular estimator of the autocorrelation function based on ranks, and show that the finite sample behaviour of this latter estimator deteriorates in the presence of isolated additive outliers with the lag of the autocorrelations. Furthermore, another important limitation of this latter robust estimator is that its asymptotic distribution is unknown and, consequently, it does not allow for the construction of significance bands.

Finally, we illustrate the results estimating the autocorrelation function of IBEX35 daily returns, quarterly US GPD growth, and monthly US inflation. In the first case, we show that when the robust estimator is used the EMH is not rejected as none of the autocorrelations are significant. In the second example, the autocorrelations are stronger when estimated using the robust estimator instead of the usual sample autocorrelations, with a more clear picture of the temporal dependence of GDP growth. Finally, the robust estimates of the autocorrelation function of inflation are larger and more persistent than those estimated using the sample autocorrelations, suggesting a weaker evidence of long-memory.

This papers opens several important topics for further research. The first obvious extension is to analyse the robust estimation of the partial autocorrelations. The second extension in our research agenda is the analysis of the results in the presence of uncorrelated but not independent noises (conditional heteroscedasticity). Note that the extension to conditionally heteroscedastic innovations should be trivial. Furthermore, it is also interesting to derive the joint asymptotic distribution of vectors of autocorrelations not only when they are zero but also in cases in which they are different from zero. Finally, the effect of outliers on persistence may be relevant for the modelling of time series.

\section*{Acknowledgements}

Financial support from the Spanish National Research Agency (Ministry of Science and Technology) Projects PID2023-150095NB-C44 (first  author) and PID2022-139614NB-C22 (second author), both funded by MCIN/AEI/10.13039/501100011033/FEDER, EU, is gratefully acknowledged. We are grateful to participants at the 19th International Conference on Computational and Financial Econometrics (London, December 2025) and in particular to C. Croux, for their comments. Any remaining errors are obviously our responsibility.

\theendnotes

 \end{document}